\documentclass[11pt]{article}
\usepackage{amsmath, amssymb, graphicx, dsfont, xcolor, hyperref, natbib, amsthm}
\usepackage{subfigure, enumerate}
\usepackage[margin=1in]{geometry}
\newtheorem{proposition}{Proposition}
\newtheorem{lemma}{Lemma}
\newtheorem{remark}{Remark}

\def\bSig\mathbf{\Sigma}

\begin{document}

\label{firstpage}

%  put the summary for your paper here

%\title[Statistical Inference for Qualitative Interactions]{Statistical Inference for Qualitative Interactions with Applications to Precision Medicine and Differential Network Analysis}
\title{Statistical Inference for Qualitative Interactions with Applications to Precision Medicine and Differential Network Analysis}
\author{Aaron Hudson and
Ali Shojaie \\
Department of Biostatistics, University of Washington, Seattle, Washington}
\date{}
\maketitle

\begin{abstract}
Qualitative interactions occur when a treatment effect or measure of association varies in sign by sub-population.
Of particular interest in many biomedical settings are absence/presence qualitative interactions, which occur when an effect is present in one sub-population but absent in another.
Absence/presence interactions arise in emerging applications in precision medicine, where the objective is to identify a set of predictive biomarkers that have prognostic value for clinical outcomes in some sub-population but not others.
They also arise naturally in gene regulatory network inference, where the goal is to identify differences in networks corresponding to diseased and healthy individuals, or to different subtypes of disease; such differences lead to identification of network-based biomarkers for diseases.
In this paper, we argue that while the absence/presence hypothesis is important, developing a statistical test for this hypothesis is an intractable problem.
To overcome this challenge, we approximate the problem in a novel inference framework.
In particular, we propose to make inferences about absence/presence interactions by quantifying the relative difference in effect size, reasoning that when the relative difference is large, an absence/presence interaction occurs. 
The proposed methodology is illustrated through a simulation study as well as an analysis of breast cancer data from the Cancer Genome Atlas.
\end{abstract}

%  Please place your key words in alphabetical order, separated
%  by semicolons, with the first letter of the first word capitalized,
%  and a period at the end of the list.
%

%\begin{keywords}
%Qualitative interactions; Differential network connectivity; Precision medicine.
%\end{keywords}

%  As usual, the \maketitle command creates the title and author/affiliations
%  display 

%  If you are using the referee option, a new page, numbered page 1, will
%  start after the summary and keywords.  The page numbers thus count the
%  number of pages of your manuscript in the preferred submission style.
%  Remember, ``Normally, regular papers exceeding 25 pages and Reader Reaction 
%  papers exceeding 12 pages in (the preferred style) will be returned to 
%  the authors without review. The page limit includes acknowledgements, 
%  references, and appendices, but not tables and figures. The page count does 
%  not include the title page and abstract. A maximum of six (6) tables or 
%  figures combined is often required.''

%  You may now place the substance of your manuscript here.  Please use
%  the \section, \subsection, etc commands as described in the user guide.
%  Please use \label and \ref commands to cross-reference sections, equations,
%  tables, figures, etc.
%
%  Please DO NOT attempt to reformat the style of equation numbering!
%  For that matter, please do not attempt to redefine anything!

\newpage

\section{Introduction}
\label{s:intro}

An objective of many biomedical studies is to identify and test for \textit{interactions}, which arise when a measure of effect or association between variables differs by sub-population.
Precision medicine and genetic network inference provide examples of areas in which interactions are of interest.
For instance, researchers in precision medicine seek to understand how patients'  characteristics are associated with heterogeneity in response to treatment.
In genetics studies, it is of interest to determine how genetic networks, which summarize the associations between genes, depend on phenotype.

Interactions may lack clinical or scientific significance when differences in effect are small.
In addition to detecting interactions, it is important to identify which are meaningful.
For example, in precision medicine, the most important differences in treatment effect may be those in which some sub-populations of patients benefit from the treatment, while other sub-populations are harmed or unaffected.
Additionally, one may want to identify differences among sub-populations in the set of biomarkers that have prognostic value for a health outcome --- that is, to determine whether some biomarkers are predictive of the outcome in only a subset of the full population.
Genetic network inference provides another example: When comparing sub-population level genetic networks, it may be of primary interest to identify pairs of genes that share an association in some sub-populations but share no association in others, or to identify pairs that have a positive association in one sub-population and a negative association another.
This is known as differential network biology \citep{ideker2012differential}.

Such \textit{qualitative interactions} are the focus of this paper.
Qualitative interactions occur when a measure of effect differs in sign by sub-population.
We consider two types of qualitative interactions: positive/negative interactions --- also known in the literature as cross-over interactions \citep{gail1985testing} --- and absence/presence interactions --- sometimes referred to as pure interactions \citep{vanderweele2019interaction}.
Positive/negative interactions occur when an effect is positive in one sub-population and negative in another, and absence/presence interactions occur when the effect is present in one population but absent in another.

Our objective is to formally test for qualitative interactions, given independent samples from each sub-population.
Testing for positive/negative interactions is well-studied \citep{gail1985testing, piantadosi1993comparison, pan1997test, silvapulle2001tests, li2006detecting}, while testing for absence/presence interactions has received substantially less attention.
Na\"ive approaches, to be discussed in the sequel, require an untenable minimum signal strength condition --- that if an effect is present in any sub-population, it is large enough to be detected with absolute certainty.
No approaches exist, to the best of our knowledge, that obviate this assumption.

In this paper, we propose a novel framework for inference about absence/presence interactions.
Our proposed methodology allows for well-calibrated hypothesis testing under mild assumptions.
We also introduce a numerical summary that measures the strength of absence/presence interactions, while accounting for the uncertainty associated with parameter estimation. 
Additionally, we describe methods for simultaneous inference about absence/presence and positive/negative interactions.
The methodology we introduce provides an effective and flexible inference tool in precision medicine and genetic network analysis, as we illustrate in simulations and an analysis of breast cancer data from The Cancer Genome Atlas (TCGA).

\section{Background}

\subsection{Notation}

As we begin to formalize the problem, we first introduce some notation.
We consider two sub-populations, labeled by $g \in \{1,2\}$.
Let $\theta_g \in \mathds{R}$ denote a measure of association in sub-population $g$.
When convenient, we write $\theta = (\theta_1, \theta_2)$.
We can consider various measures of association, such as: correlation coefficients, indicating the strength of  linear relationship between two variables of interest; log odds ratios, describing the relationship between predictors and a binary outcome; and log hazard ratios, describing the association between predictors and a time-to-event outcome. 

We assume that given i.i.d. samples of size $n_1$ and $n_2$ from each sub-population, $\sqrt{n_g}$-consistent and asymptotically normal estimates $\hat{\theta}_g$ of $\theta_g$ are available, i.e.,
\[
\sqrt{n_g}\left(\hat{\theta}_g - \theta_g\right) \to_d N\left(0, {\sigma_g^2}\right),
\]
with ${\sigma^2_g} > 0$ denoting the asymptotic variance. 
For expositional simplicity, we assume balanced sample sizes $n_1 = n_2 = n$ (key results are stated more generally in the Appendix).
We also assume ${\sigma^2_g}$ is known, though we can instead use a consistent estimate, as is commonly done in practice.

We now formally state the null hypotheses of no positive/negative interactions and no absence/presence interaction, labeled $H_0^{\text{P/N}}$ and $H_0^{\text{A/P}}$, respectively:
\begin{align}
&H_0^{\text{P/N}}: \theta_g \geq 0 \text{ for } g \in \{1,2\}  
\textbf{ or }
\theta_{g} \leq 0 \text{ for } g \in \{1,2\} 
\\
&H_0^{\text{A/P}}: \theta_g = 0 \text{ for } g \in \{1, 2\}
\textbf{ or }
\theta_{g} \neq 0 \text{ for } g \in \{1, 2\}.
\end{align}
We let $\Theta_0^{\text{P/N}}, \Theta_0^{\text{A/P}}$ and $\Theta_1^{\text{P/N}}, \Theta_1^{\text{A/P}}$ denote the corresponding null and alternative regions of the parameter space, depicted in Figure \ref{fig:null}. (Recall that the null region is the set of parameters such that the null hypothesis holds, and the alternative region is the complement of the null region.)
The positive/negative null region is the union of the the non-negative and non-positive orthants, and the absence/presence null region is the union of all open orthants and the origin.

\begin{figure}[!h]
\centering
\includegraphics{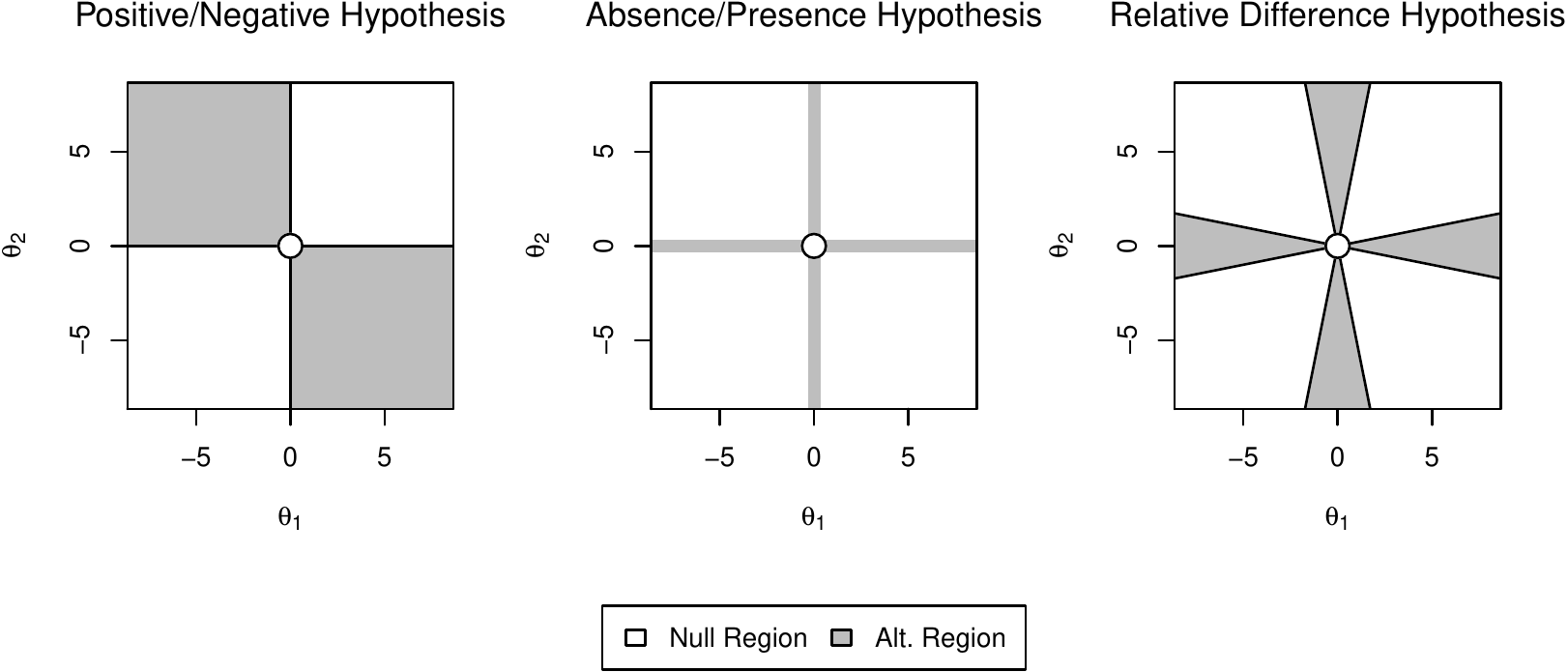}
\caption{Null and alternative regions of parameter space for positive/negative, absence/presence, and relative difference hypotheses.}
\label{fig:null}
\end{figure}

\subsection{Testing Composite Null Hypotheses}

Our goal is to use the estimate $\hat{\theta}$ to perform tests of $H_0^{\text{P/N}}$ and $H_0^{\text{A/P}}$ such that the size is controlled asymptotically under mild assumptions.
Recall that for a null hypothesis $H_0$ with accompanying null region $\Theta_0$, the size of a test is defined as
\[
\sup_{\theta_0 \in \Theta_0} \mathbb{P
}(\text{``Reject the null hypothesis''}| \theta = \theta_0).
\]
In words, the size is the largest possible type-I error rate that could be achieved under any probability distribution given that $\theta$ belongs to the null region.

Here, we describe our general approach for controlling the size at a pre-specified level $\alpha \in (0,1)$.
We first define a test statistic $T$, a map from observable data to a real-valued number, with larger values of $T$ corresponding to more evidence against the null hypothesis.
We then calculate the test statistic on the observed data, which we denote by $t$. 
We write $\mathbb{P}(T > t|\theta = \theta_0)$ as the probability of observing a random test statistic at least as large as the observed value $t$, assuming $\theta = \theta_0$.
We reject the null hypothesis if
\[
\rho(t) \equiv \sup_{\theta_0 \in \Theta_0} \mathbb{P}(T > t|\theta = \theta_0) < \alpha.
\]
One can think of $\mathbb{P}(T > t|\theta = \theta_0)$ as the p-value under a specific null distribution $\theta = \theta_0$; $\rho(t)$ is then the largest of all such p-values.
We reject $H_0$ when there is sufficient evidence to reject all hypotheses $\theta = \theta_0$.
We can view $\rho(t)$ as a generalization of the usual p-value for simple null hypotheses to tests with composite null hypotheses, and will simply refer to $\rho(t)$ as ``p-value''.
Tests of the above form are guaranteed to control the size \citep{casella2002statistical}.

\subsection{Existing Methodology}

We now review existing approaches to test for qualitative interactions. We first discuss testing positive/negative interactions before moving to absence/presence interactions.

\cite{gail1985testing} developed the most widely used procedure to test for positive/negative interactions.
Though Gail and Simon proposed a general $K$-sample test, we focus on the two-sample problem in this paper.
We note that various $K$-sample tests for positive/negative interactions have been proposed \citep{piantadosi1993comparison, silvapulle2001tests, li2006detecting}, but these procedures are essentially equivalent to the Gail-Simon test in the two-sample setting.

Gail and Simon's approach is to perform a likelihood ratio test based on the asymptotic sampling distribution of $\hat{\theta}$.
The likelihood ratio test rejects $H_0^{\text{P/N}}$ for large values of
\begin{align*}
-\,
\frac{\sup_{(a_1, a_2) \in \Theta_0^{\text{P/N}}}\prod_{g\in \{1,2\}} \sigma_g^{-1}\phi\left\{\sqrt{n}\sigma_g^{-1}(\hat{\theta}_g - a_g)\right\}}{\sup_{(b_1, b_2) \in \mathds{R}^2}\prod_{g\in \{1,2\}} \sigma_g^{-1}\phi\left\{\sqrt{n}\sigma_g^{-1}(\hat{\theta}_g - b_g)\right\}},
\end{align*}
where $\phi(\cdot)$ is the standard normal density.
By performing algebreic manipulations, one can show that the likelihood ratio test equivalently rejects the null for large values of
\begin{align*}
T^{\text{P/N}} = 
\min_{(a_1, a_2) \in \Theta_0^{\text{P/N}}}  \sum_{g \in \{1,2\}} 
n \left\{\sigma_g^{-1}\left(\hat{\theta}_g - a_g\right)\right\}^2 .
\end{align*}
The test statistic $T^{\text{P/N}}$ can be interpreted as the shortest distance between $\hat{\theta}$ and the null region, where the distance is inversely weighted by the asymptotic variances of the estimates, as illustrated in Figure \ref{fig:proj}.

\begin{figure}[!h]
\centering
\includegraphics[scale=.85]{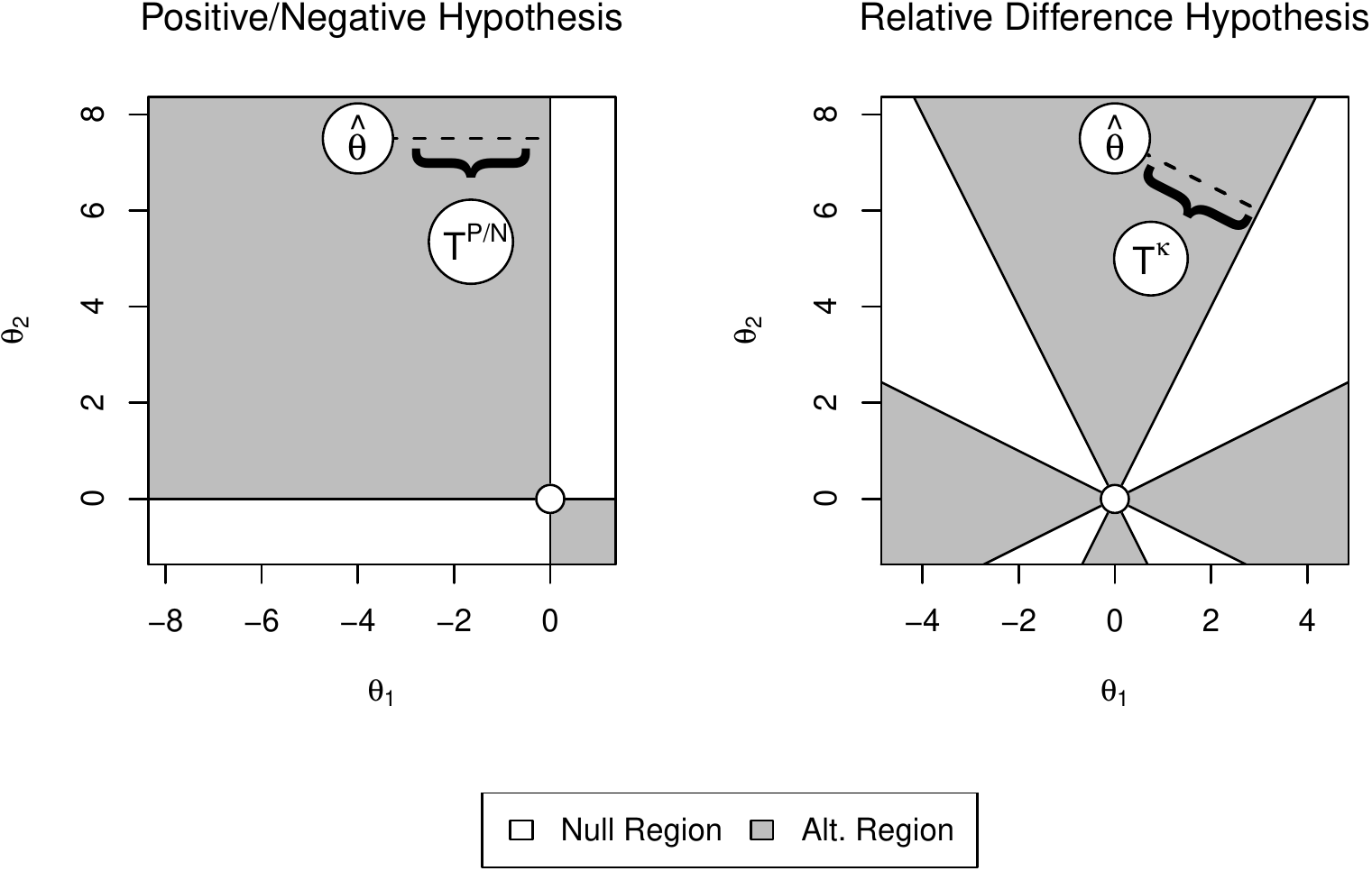}
\caption{Geometric interpretation of likelihood ratio statistic for positive/negative and absence/presence hypotheses.}
\label{fig:proj}
\end{figure}

%Allowing $\rho_{\theta_0}^{P/N}(t) \equiv P(T_n^{P/N} > t| \theta = \theta_0)$ and $t$ to be test statistic calculated from the observed data, the size $\alpha$ likelihood ratio tests rejects when
%\begin{align*}
%\bar{\rho}^{P/N}(t) \equiv \sup_{\theta_0 \in \Theta_0^{P/N}}\rho_{\theta_0}^{P/N}(t)  < \alpha
%\end{align*}
Gail and Simon show that the test statistic $T^{\text{P/N}}$ can be calculated as
\[
T^{\text{P/N}} = 
\min_{g \in \{1,2\}}\left\{
n
\left(\hat{\theta}_g/{\sigma_g}\right)^2
\right\}
\mathds{1}
\left(
\hat{\theta} \in \Theta_1^{\text{P/N}}
\right),
\]
where $\mathds{1}(\cdot)$ denotes the indicator function.
Furthermore, one can verify that for all $t > 0$, the p-value is easily calculated as
\[
{\rho}^{\text{P/N}}(t) = \sup_{\theta_0 \in \Theta_0^{\text{P/N}}} \lim_{n\to\infty} \mathbb{P}\left(T^{\text{P/N}} > t|\theta = \theta_0\right) = \frac{1}{2}\mathbb{P}\left(\chi^2_1 > t\right).
\]
The likelihood ratio test is quite intuitive and rejects the null when a positive estimate of association is observed in one population, a negative estimate is observed in another population, and both associations are statistically significant.

%Though we focus on the two-sample likelihood ratio test here, the test can be extended to a more general $k-$sample test easily.
%Many alternatives to the likelihood ratio test have been proposed for simultaneous inference with $k$-sample \textcolor{red}{(List papers here)}, though all are essentially equivalent in the two-sample case.

Now, we discuss approaches to test for absence/presence interactions.
While one might be tempted to perform a likelihood ratio test for absence/presence interactions, the likelihood ratio test fails in the sense that it \textit{never} can reject the null.
To see this, we first recognize that, similar to the positive/negative interaction test, the absence/presence likelihood ratio test would reject for large values of
\begin{align*}
T^{\text{A/P}} = 
\min_{(a_1, a_2) \in \Theta_0^{\text{A/P}}} \sum_{g \in \{1,2\}} 
n
\left\{\sigma_g^{-1}\left(\hat{\theta}_g - a_g\right)\right\}^2.
\end{align*}
Again, the test statistic $T^{\text{A/P}}$ is the shortest distance between $\hat{\theta}$ and the null region $\Theta_0^{\text{A/P}}$.
Because the alternative region $\Theta_1^{\text{A/P}}$ has zero area, $\hat{\theta}$ lies in the null region with probability one.
Therefore, the test statistic is \textit{always} 0, and the likelihood ratio test has no power.

One might alternatively attempt to test the absence/presence null by separately testing the null hypotheses $H_0^{g}: \theta_g = 0$ for $g \in \{1,2\}$ and rejecting the absence/presence null when $H_0^{g}$ is rejected for one $g$ and not rejected for the other.
To control the size of a test of this form, we need to simultaneously control the type-I error under two scenarios: (1) there is an association in both sub-populations, and (2) there is no association in either population.
When there is an association in both populations, a type-I error occurs when we incorrectly \textit{fail to reject} one of $H_0^{g}$.
When there is no association in either population, we make a type-I error when we incorrectly reject one of $H_0^{g}$.
Thus, controlling the size of the test for $H_0^{\text{A/P}}$ using this approach requires simultaneous control of the type-I error rate and type-II error rate of tests for $H_0^{g}$.
If the tests for $H_0^{g}$ are consistent --- that is, the type-II error rates tend to zero with sufficiently large samples --- this approach is asymptotically valid, as only the type-I error rates for tests of $H_0^g$ need to be controlled.
%Specifically, we require the signal to be larger than $o(n^{-1/2})$, which is known as the ``beta-min'' condition in the penalized regression literature.
However, with even moderately large samples, we will not be able to correctly reject false $H_0^g$ with absolute certainty unless the true association is strong and hence easy to detect.
This would make the test of $H_0^{\text{A/P}}$ unreliable in the presence of weak signal.
%For instance, suppose there is a moderate effect in both sub-populations, and for our observed sample size, the probability of rejecting H_0^{g} is 1/2.
%This would result in an inflated type-I error rate of 1/4, which is unnacceptable.

Specifically, we require $\theta_1, \theta_2 > o(n^{-1/2})$.
To see this, we construct a more formal argument.  
For simplicity, suppose $\sigma_1 = \sigma_2 = 1$. 
We consider tests of $H_0^g$ of the form 
\[
\psi_g = 
\begin{cases}
\text{Reject;} & \text{if } \sqrt{n} |\hat{\theta}_g| > a
\\
\text{Accept;} & \text{if } \sqrt{n} |\hat{\theta}_g| < a 
\end{cases},
\]
where $a$ is a constant that would be selected to control the size.
The probability of rejecting the absence/presence null is
\[
\mathbb{P}\left(\text{Reject } H_0^{\text{A/P}}\right)  = 
\mathbb{P}\left(\psi_1 = \text{Reject}\right)\mathbb{P}\left(\psi_2 = \text{Accept}\right) +
\mathbb{P}\left(\psi_1 = \text{Accept}\right)\mathbb{P}\left(\psi_2 = \text{Reject}\right).
\]
Suppose $\theta_1 $ and $\theta_2$ are guaranteed to be greater than $o(n^{-1/2})$ if they are both nonzero.
Then, for $\theta_1, \theta_2 \neq 0$, $\sqrt{n}|\hat{\theta}_g| \to \infty$, and  $\mathbb{P}\left(\text{Reject } H_0^{\text{A/P}} \right) \to 0$.
Therefore, to control the size, we are only required to select $\alpha$ so that the type-I error is controlled when $\theta_1 = \theta_2 = 0$; this can be done by taking $a$ as the $(1-\alpha/4)$ quantile of the standard normal distribution.
However, if we allow $\theta_g < o(n^{1/2})$, we can see a drastically inflated type-I error rate.
For instance, for a small $\epsilon > 0$, let $\theta_1 = n^{-1/2 + \epsilon}$, $\theta_2 = n^{-1/2 - \epsilon}$.
Then $\sqrt{n}|\hat{\theta}_1| \to \infty > a$ while $\sqrt{n}|\hat{\theta}_2| \to 0 < a$, so $\mathbb{P}\left(\text{Reject } H_0^{\text{A/P}}\right) \to 1$.
Thus, when small signal is permitted, tests of this form will be asymptotically anti-conservative.
%We conclude this remark by noting that the minimum signal strength condition for testing consistency is essentially the same as the so-called ``$\beta$-min" condition that LASSO-type methods require for precise variable selection \citep{meinshausen2006high, zhao2006model}.

Both approaches discussed above for testing absence/presence interactions fail for a similar reason: it is difficult to gather evidence supporting that a measure of association is exactly equal to zero.
This is captured by the alternative region having zero area, causing the failure of the first approach.
In the second approach, to obtain evidence supporting that an association is zero, we require that $H_{0,g}$ is only accepted when $\theta_g = 0$; for this, we rely upon a minimum signal strength condition to guarantee that any non-zero association is detected.

\section{Proposed Methodology}

\subsection{Refinement of Absence/Presence Hypothesis}

%The absence/presence null hypothesis is difficult to test.
%We find it somewhat inappropriate to state that under the alternative hypothesis, one measure of association will be exactly equal to zero because we never expect this to be true.
%Rather, we expect that an association may be negligible or \textit{close} to zero.
%Therefore, we find it more reasonable to say that an absence/presence occurs when one measure association is moderately large while the other is very small, or when one is much larger than the other.

To mitigate the challenges described in Section 2, we consider a refinement of the absence/presence null hypothesis.
The key idea is that in practice, absence/presence interactions can be approximated by considering the settings where an association is at least moderately large in one population and negligible or \textit{near} zero in the other; or when one association is substantially stronger than the other.
This means that we can expand the alternative region to include neighborhoods of zero in a way that the absence/presence interpretation is preserved.

Recall that when there exists an absence/presence interaction, the ratio of the maximum of the absolute value of the $\theta_g$ to the minimum is infinite.
We cannot test that the ratio is infinite because we will never have evidence to support that the denominator is exactly zero.
However, we can test that the ratio is \textit{large} because we may have evidence to support that the denominator is very small.
Motivated by this intuition, we propose to test whether the relative difference between sub-population measures of association is greater than a large pre-specified constant $\kappa > 1$.
Formally, let $\theta_{\max} = \max_{g}|\theta_g|$ and $\theta_{\min} = \min_{g}|\theta_g|$.
We define the new relative difference null hypothesis $H_0^{\kappa}$ as
\[
H_0^{\kappa}: \theta_{\max}/\theta_{\min} \leq \kappa  \text{ or } \theta_{\max} = \theta_{\min} = 0.
\]
Equivalently,
\begin{align}
H_0^\kappa: \theta_{\max} - \kappa \theta_{\min} \leq 0. \label{eq:rdnull}
\end{align}
The null region $\Theta_0^{\kappa}$, illustrated in Figure \ref{fig:null}, is the union of four linear subspaces --- each residing in a separate orthant of $\mathds{R}^2$; the boundary of each subspace is the union of the spans of vectors with absolute direction $(\kappa, 1)'$ and $(1, \kappa)'$.
%The boundary of the null region is the union of the spaces spanned by $(\kappa, 1)'$, $(1, \kappa)'$, $(-1, \kappa)'$, and $(1, -\kappa)'$.

The relative difference null region can be viewed as a relaxation of the absence/presence null.
For a large choice of $\kappa$, both our original and refined null hypotheses have the same interpretation: the greater measure of association is substantially larger than the lesser.
However, we find it appealing that $H_0^{\kappa}$ has a reasonable interpretation for any choice of $\kappa$; that is, the multiplicative difference in strength of association is no larger than $\kappa$.

To motivate defining the refined null hypothesis in terms of relative differences rather than absolute differences, we argue that testing for relative differences has at least the following benefits.
First, relative differences are unitless, so the relative difference null is compatible with unitless measures of association such as the Pearson correlation coefficient, which are often preferred in the analysis of biological data.
Second, a reasonable $\kappa$ can be selected without prior knowledge of ranges of strength of association.

Of course, the relative difference null hypothesis depends on the choice of $\kappa$.
For small values of $\kappa$, the relative difference null may be too dissimilar from the absence/presence null to retain its interpretation; for large $\kappa$, the alternative region becomes very small, and the hypothesis may be overly conservative.
In the following subsections, we first construct a test of the relative difference null hypothesis for a pre-specified $\kappa$  and then describe an approach to identify the set of $\kappa$ such that the test rejects the null hypothesis, as to circumvent tuning parameter selection.
%We emphasize that the choice of $\kappa$ is left to the practitioner.

\subsection{Likelihood Ratio Test for Relative Difference Hypothesis}

We now develop a testing procedure for the new relative difference null hypothesis.
The relative difference null region, unlike the absence/presence null region, has non-zero area, so a likelihood ratio test will not fail in the same manner as the likelihood ratio test for absence/presence interactions.
Similar to the previously discussed examples, the likelihood ratio test statistic is
\begin{align*}
T^{\kappa} &=
\min_{(a_1, a_2) \in \Theta_0^{\kappa}}
\sum_{g \in \{1,2\}} 
n\left\{\sigma_g^{-1}\left(\hat{\theta}_g - a_g\right)\right\}^2,
\end{align*}
and can be interpreted as the shortest (weighted) distance between $\hat{\theta}$ and the null region $\Theta_0^\kappa$.
Clearly, the test statistic is zero whenever $\hat{\theta}$ lies in the null region.
Otherwise, $T^{\kappa}$ is the shortest distance between $\hat{\theta}$ and the closest of the four linear subspaces that define $\Theta_0^\kappa$. 
The test statistic can be calculated as the distance between $(\hat{\theta}_{\max}, \hat{\theta}_{\min})'$ and its projection onto the span of the vector $(\kappa , 1)'$.
The test statistic's geometric interpretation is illustrated in Figure \ref{fig:proj}.

%The boundary of the null region is the union of the spaces spanned by $(\kappa, 1)'$, $(1, \kappa)'$, $(-1, \kappa)'$, and $(1, -\kappa)'$.
%Of the four vector spaces, $\hat{\theta}$ is closest to the one which shares the same sign and has maximum and minimum elements in the same order.

The likelihood ratio test statistic is straightforward to calculate. 
Let $\hat{\theta}_{\max} = \max_g|\hat{\theta}_g|$ and $\hat{\theta}_{\min} = \min_g|\hat{\theta}_g|$ be the strongest and weakest estimated absolute association.
%Further, let $\hat{\tau}_{\max} = n^{-1}\sigma^2_1\mathds{1}\left(|\hat{\theta}_1| = \hat{\theta}_{\max}\right) + n^{-1}\sigma^2_2\mathds{1}\left(|\hat{\theta}_2| = \hat{\theta}_{\max}\right)$ and $\hat{\tau}_{\min} = n^{-1}\sigma^2_1\mathds{1}\left(|\hat{\theta}_1| = \hat{\theta}_{\min}\right) + n^{-1}\sigma^2_2\mathds{1}\left(|\hat{\theta}_2| = \hat{\theta}_{\min}\right)$ be the asymptotic variances corresponding to the strongest and weakest estimated associations.
In the following lemma, we state that $T^{\kappa}$ is equal to the difference between $\hat{\theta}_{\max}$ and $\kappa\hat{\theta}_{\min}$ divided by a normalizing constant.
% $n^{-1}(\hat{\sigma}_{\max}^2 + \kappa^2 \hat{\sigma}_{\min}^2)$.
Thus, the test statistic can be viewed as a plug-in of $\hat{\theta}$ into \eqref{eq:rdnull} with an additional normalizing constant and closely resembles the test statistic proposed by \cite{fieller1940biological} to conduct inference about ratios of means.

\begin{lemma}
The likelihood ratio test statistic $T^{\kappa}$ can be written as
\[
T^{\kappa} = 
\frac{\hat{\theta}_{\max} - \kappa\hat{\theta}_{\text{min}}}{\hat{\tau}_{\max} + \kappa^2\hat{\tau}_{\min}},
\]
where $\hat{\tau}_{\max} = n^{-1}\sigma^2_1\mathds{1}\left(|\hat{\theta}_1| = \hat{\theta}_{\max}\right) + n^{-1}\sigma^2_2\mathds{1}\left(|\hat{\theta}_2| = \hat{\theta}_{\max}\right)$, and $\hat{\tau}_{\min} = n^{-1}\sigma^2_1\mathds{1}\left(|\hat{\theta}_1| = \hat{\theta}_{\min}\right) + n^{-1}\sigma^2_2\mathds{1}\left(|\hat{\theta}_2| = \hat{\theta}_{\min}\right)$.
\end{lemma}

We now discuss how to obtain a p-value for the relative difference hypothesis.
First, we obtain an observed test statistic $t$, a realization of $T^\kappa$ calculated from the data.
Following the approach described in Section 2.2, we define the p-value as $\rho^\kappa(t)$, the largest of all asymptotic tail probabilities $\lim_{n \to \infty}\mathbb{P}(T^\kappa > t| \theta = \theta_0)$ such that $\theta_0$ belongs to the null region.
To determine the maximum tail probability, we characterize the limiting distribution of $T^\kappa$ assuming $\theta = \theta_0$ for all $\theta_0$ in the null region.

Though the null region contains an infinite number of values, $T^\kappa$ can only attain one of three limiting distributions corresponding to the following three cases:
\begin{enumerate}
\item The true association is in the interior of the null region, i.e., $\theta_{\max} - \kappa \theta_{\min} < 0$.
\item The true association is on the boundary of the null region, but both associations are non-zero, i.e., $\theta_{\max} = \kappa \theta_{\min} > 0 $.
\item The true association is zero in both sub-populations, i.e., $\theta_1 = \theta_2 = 0$.
\end{enumerate}

In Proposition 1, we describe the asymptotic behavior of $T^\kappa$ for cases 1 and 2 above.
We provide here some intuition for the result and reserve a formal argument for the Appendix
In case 1, it is easy to argue that because $\hat{\theta}$ is consistent, $T^\kappa$ is a negative number with probability tending to one, and therefore never provides evidence against the null.
In case 2, $T^\kappa$ asymptotically follows a standard normal distribution.
To see this, we note that because both associations are non-zero, consistency and asymptotic normality of $\hat{\theta}$ imply that, for large $n$, the sign and order of the estimates are deterministic.
That the signs are asymptotically deterministic implies that $|\hat{\theta}|$ is asymptotically normal (speaking loosely, $|\hat{\theta}_g| \to \text{sign}(\theta_g)\hat{\theta}_g$), and that order is asymptotically deterministic implies that $\hat{\theta}_{\max}$ and $\hat{\theta}_{\min}$ are asymptotically independent.
Therefore, taking the difference between $\hat{\theta}_{\max}$ and $\hat{\theta}_{\min}$, suitably standardized, is asymptotically equivalent to taking a difference between two independent normal random variables with equal means.
Dividing by the asymptotic variances gives the claimed result.

\begin{proposition}
\textit{
In the interior of the null region, i.e., when $\theta_{\max} - \kappa \theta_{\min} < 0$, $T^\kappa$ converges in distribution to $-\infty$.  At all nonzero boundary points of the null region, i.e., when $\theta_{\max} = \kappa \theta_{\min} > 0$, $T^\kappa$ converges in distribution to a standard normal random variable.}
\end{proposition}

In case 3, when both associations are zero, the asymptotic distribution of the test statistic is more complicated.
More specifically, $\theta_g = 0$ implies that, asymptotically, $|\hat{\theta}_g|$ follows a half-normal distribution instead of a normal distribution.
Moreover, the order of $|\hat{\theta}|$ remains random in the limit.
This gives rise to a non-standard limiting distribution.
In particular, unlike cases 1 and 2, the limiting distribution depends on the asymptotic variances of the sub-population estimates (and also the ratio of the sample sizes in the unbalanced case).
We are nonetheless able to derive an analytic expression for the distribution function, stated in Proposition 2.
We reserve this statement for the Appendix, as the expression is cumbersome.

By Propositions 1 and 2, we can calculate the p-value as the maximum of the tail probabilities in cases 2 and 3.
That is,
\begin{align}
\rho^\kappa(t) = \max\left\{1-\Phi(t), 1 - F^\kappa(t)\right\},
\label{eq:rdpval}
\end{align}
where $\Phi(\cdot)$ denotes the standard normal distribution function and $F^\kappa(\cdot) \equiv \lim_{n \to \infty} \mathbb{P}(T^\kappa < t|\theta = \theta_0)$ is the limiting distribution function for the test statistic when $\theta = 0$.
Though the limiting distribution under $\theta = 0$ is non-standard, tail probabilities can be calculated easily, as we show in the Appendix (Remark 1).  The p-value is therefore simple to calculate.

%We note that if in \eqref{eq:rdpval}, $1 - F^\kappa(t)$ is replaced by the p-value resulting from any valid test of $\theta = 0$, we still arrive at a valid test.
%That is, to reject the null, we only require evidence that $\theta \neq 0$ and that $\theta_{\max} - \kappa\theta_{min} > 0$ given that $\theta \neq 0$.
%One may expect to find a more powerful test of $H_0^\kappa$ by substituting  $1 - F^\kappa(t)$ with, say, the p-value from a most powerful test of $\theta = 0$.
%However, it is uncommon in practice that $1 - \Phi(t) <  \alpha < 1 - F^\kappa(t)$, so the effect of replacing $1 - F^\kappa(t)$ is usually small.

We have derived an analytic approximation for the power of the likelihood ratio test.
In Proposition 2, we more generally characterize the limiting distribution of the test statistic under hypotheses of the form $(\theta_1, \theta_2) = n^{-1/2}(c_1, c_2)$.
The local asymptotic power of the proposed test at the level $\alpha$ is available by considering $\beta^\kappa_\alpha(c_1, c_2) \equiv \lim_{n \to \infty} \mathbb{P}\left(T^\kappa > t^*_{1-\alpha} | \theta = n^{-1/2}(c_1, c_2)\right)$, where we define $t^*_{1 - \alpha}$ as the maximum of the $1-\alpha$ quantiles of the limiting distribution of $T^\kappa$ under scenarios 2 and 3 described above. 
An analytic finite-sample approximation of the power can be calculated as $\beta^\kappa_{\alpha}\left(n^{1/2}\theta_1, n^{1/2}\theta_2\right)$.

A contour plot of the local asymptotic power is given in Figure \ref{fig:power}.
We consider both cases of equal and unequal variance of the estimators.
We observe that the likelihood ratio test has low power when the strongest effect $\max\left\{|c_1|, |c_2| \right\}$ is small, and power improves considerably when the strongest effect grows.
Additionally, we find that in the presence of unequal variance, the test has greater power when the weakest effect is estimated with higher precision than the strongest effect.

\begin{figure}[!h]
\centering
\includegraphics{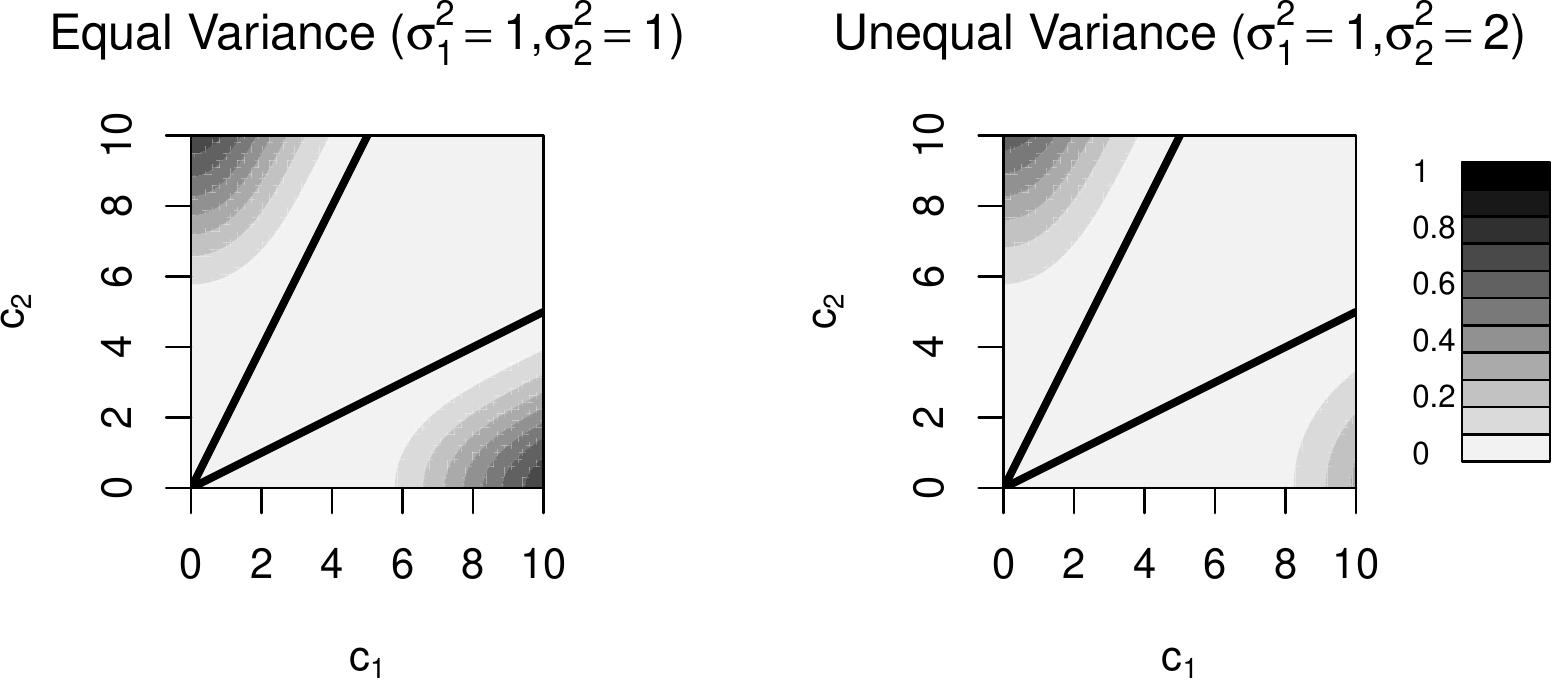}
\caption{Contour plot of local asymptotic power of relative difference likelihood ratio test with $\kappa = 2$ and $\alpha = .05$.   Bold lines represent the boundary of the null region. Settings of equal and unequal asymptotic variance of estimators are represented. }
\label{fig:power}
\end{figure}

\subsection{Quantifying Relative Difference in Effect by Inverting Likelihood Ratio Test}

Rather than perform a the likelihood ratio test for a pre-specified $\kappa$, one may prefer to directly estimate the relative difference in effect.
Na\"ively, one might consider estimating the relative difference in effect as $\hat{\theta}_{\max}/\hat{\theta}_{min}$.
However, this estimate will behave poorly when $\theta_1 = \theta_2 = 0$ and should therefore not be reported in practice. 
%\citep{franz2007ratios}.

To overcome this issue, we propose to quantify the relative difference in effect size by inverting the likelihood ratio test, similar to \cite{fieller1940biological}.
We define 
\[
\kappa_{\max}^\alpha \equiv \sup\{\kappa: \kappa > 1, \rho^\kappa(t) < \alpha\}
\] 
as the largest $\kappa > 1$ such that the likelihood ratio test rejects the null hypothesis at the $\alpha$ level.
When the likelihood ratio test fails to reject for all $\kappa > 1$, we will use the convention $\kappa_{\max}^{\alpha} = 1$.
We find it appealing that $\kappa^\alpha_{\max}$ converges to $1$ if $\theta_{\max} = \theta_{\min}$, and $\kappa_{\max}$ should approach but not exceed $\theta_{\max}/\theta_{\min}$ otherwise.
We discuss calculation of $\kappa^\alpha_{\max}$ in the Appendix.

\subsection{Simultaneous Test of Qualitative Interactions}

When it is of interest to identify both absence/presence and positive/negative qualitative interactions, it may be desirable to test for both simultaneously.
In this section, we construct an omnibus test that achieves exact control of the size asymptotically.
  
We define the omnibus qualitative interaction null hypothesis as
\[ 
H_0^{\text{P/N},\kappa}: \text{ Both } H_0^{\text{P/N}} \text{ and } H_0^{\kappa} \text{ hold.} 
\]
The null region $\Theta_0^{\text{P/N}, \kappa}$ is the intersection of the positive/negative and relative difference null regions, as depicted in Figure \ref{fig:projOmnibus}.
We observe that as $\kappa \to \infty$, the omnibus null and alternative regions tend to the positive/negative null and alternative regions.

\begin{figure}[!h]
\centering
%\textbf{Web Figure A: }\par\medskip
\includegraphics[scale = .825]{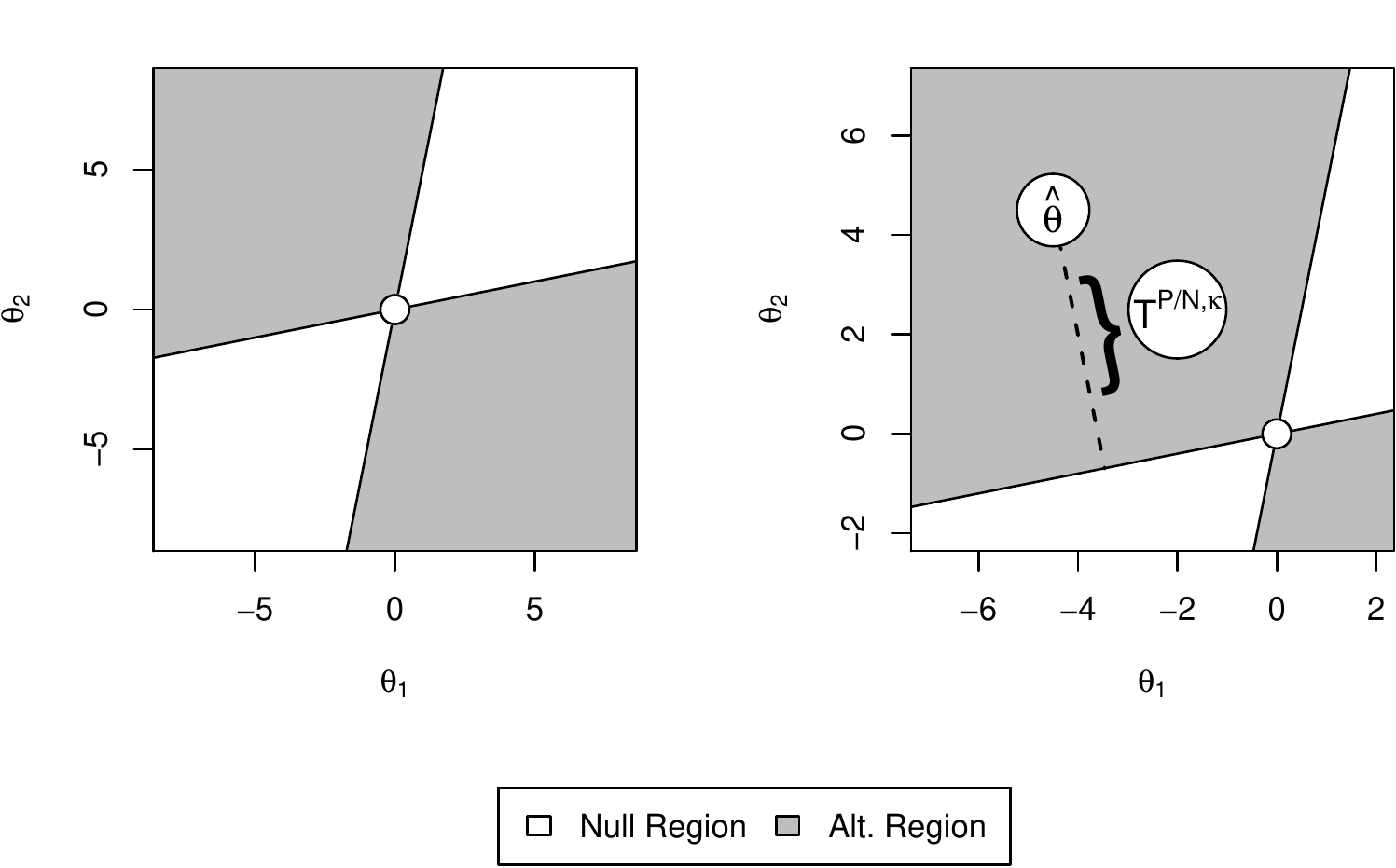}
\caption{(Left) Null and alternative region for omnibus qualitative interaction hypothesis. (Right) Geometric interpretation of likelihood ratio statistic for omnibus qualitative interaction hypothesis.}
\label{fig:projOmnibus}
\end{figure}

To construct the likelihood ratio test, we proceed using similar arguments to those presented in Section 3.2.  
The likelihood ratio statistic $T^{\text{P/N},\kappa}$ is the distance between the estimate $\hat{\theta}$ and its projection onto the null region $\Theta_0^{\text{P/N}, \kappa}$, inversely weighted by the asymptotic variance of $\hat{\theta}$.
A simple expression for $T^{\text{P/N}, \kappa}$ is given in Lemma 2.

\begin{lemma}
The likelihood ratio statistic $T^{\text{\normalfont{P/N}}, \kappa}$ can be written as
\[
T^{\text{\normalfont{P/N}},\kappa} =
\min\left\{
\frac{\left(\hat{\theta}_1 - \kappa \hat{\theta}_2\right)^2}{n^{-1}(\sigma_1^2 + \kappa^2\sigma^2_2)}
,
\frac{\left(\kappa \hat{\theta}_1 - \hat{\theta}_2\right)^2}{n^{-1}(\kappa^2\sigma_1^{2} + \sigma_2^2)}
\right\}
\mathds{1}\left(\hat{\theta} \in \Theta_1^{\text{\normalfont{P/N}},\kappa}\right).
\]
\end{lemma}
Unsurprisingly, the likelihood ratio statistic for the omnibus test approaches the likelihood ratio statistic for the Gail-Simon likelihood ratio statistic for positive/negative interactions in the limit of large $\kappa$.
The tests will, therefore, be nearly identical for sufficiently large $\kappa$.

To characterize the asymptotic behavior of the omnibus test statistic at each location under the null, we use similar arguments to those in the previous section.
If $\theta$ belongs to the interior of the null region, $T^{\text{P/N},\kappa}$ converges in probability to zero.
If $\theta$ belongs to the boundary of the null region and is non-zero, $T^{\text{P/N},\kappa}$ converges weakly to a uniform mixture of zero and the chi-squared distribution with one degree of freedom.
If $\theta$ is zero, the limiting distribution of $T^{\text{P/N},\kappa}$ is non-standard, though it can be characterized nonetheless.
Formal statements of asymptotic properties of $T^{\text{P/N},\kappa}$ are given in Propositions 3 and Remark 2 (the statement of Remark 2 is also cumbersome, and is reserved for the Appendix).
%, and accompanying proofs are provided in the Appendix.

\setcounter{proposition}{2}

\begin{proposition}
\textit{
If $\theta$ belongs to the interior of the null region $\Theta_0^{\text{P/N},\kappa}$, $T^{\text{\normalfont{P/N}}, \kappa}$ converges in distribution to zero.
If $\theta$ is on the boundary of the null region, but $\theta \neq 0$, $\mathbb{P}(T^{\text{\normalfont{P/N}}, \kappa} > t) \to \frac{1}{2}P\left(\chi^2_1 > t\right)$ as $n \to \infty$.
}
\end{proposition}

Calculating the p-value for the omnibus test is no more difficult than calculating the p-value for the absence/presence test.
Defining $F^{\text{P/N},\kappa}(t) \equiv \lim_{n \to \infty}\mathbb{P}(T^{\text{P/N}} < t|\theta = 0)$, the p-value can be calculated as
\begin{align}
\rho^{\kappa, \text{P/N}}(t) = \max\left\{\mathbb{P}\left(\chi^2_1 > t\right), 1 - F^{\text{P/N},\kappa}(t)  \right\},
\label{eqn:OMNpval}
\end{align}
where $t$ is the value of  the test statistic $T^{\text{P/N},\kappa}$ calculated on the observed data.
%We can also obtain a valid p-value by, in \eqref{eqn:OMNpval}, replacing $1 - F^{\text{P/N},\kappa}(t)$ with the p-value obtained from any test of $\theta = 0$, similarly as discussed in Section 3.3.
We characterize the local asymptotic power of the omnibus likelihood ratio test in Proposition 4 in the Appendix.

\section{Simulation Study}

In a Monte Carlo simulation study, we examine how type-I error probabilities and statistical power of the likelihood ratio tests for qualitative interactions are affected by signal strength, sample size, and selection of $\kappa$.
Additionally, we examine how $\kappa^\alpha_{\max}$ depends on the true sub-population effects and the sample size.

We generate random observations $(Y_g, X_g)$ in sub-population $g$ under the linear model:
\[
Y_{g} = \theta_gX_g + \epsilon;\, X_g \sim N(0,1);\, \epsilon \sim N(0,1)
.\]
Here, $X_g$ is the predictor of interest, $Y_g$ is the response, and $\epsilon$ is white noise.
The measure of association in which we are interested is the regression coefficient $\theta_g$.
We fix $\theta_1 = 1$ and consider $\theta_2 \in \{-1,-.9,....,.9,1\}$.
A total of $1000$ synthetic data sets are randomly generated for each $\theta_2$ and $n \in \{50, 100\}$.

For each synthetic data set, we perform both the relative difference likelihood ratio test and the omnibus qualitative interaction likelihood ratio test with $\kappa = 2$ and $\kappa = 4$; we use a significance level of $\alpha = .05$.
We additionally calculate $\kappa^\alpha_{\max}$ with $\alpha = .05$.
Parameter estimation is performed with ordinary least squares, and model-based estimates of the standard error are used.

In Figure \ref{fig:PowerSim}, we plot the Monte Carlo estimate of the rejection probability of the relative difference likelihood ratio test over the range of $\theta_2$.
We see that the test achieves control of size for both choices of $\kappa$ and both sample sizes.
We observe that power is largest when $|\theta_2|$ is near zero, as we would expect.
Power is moderately strong when $\kappa = 2$ but fairly poor when $\kappa = 4$.
With $\kappa = 4$, the likelihood ratio test has almost zero power when the sample size is small, though we see modest improvement when the sample size is larger.

Figure \ref{fig:PowerSim} also shows the estimated rejection probabilities of the omnibus test over $\theta_2$.
Control of size is achieved in both large and small samples, as expected.
We note that for the omnibus test, power increases as $\theta_2$ tends to $-1$, and power is uniformly larger when $\kappa = 2$ than when $\kappa = 4$.

\begin{figure}[!h]
\centering
\includegraphics{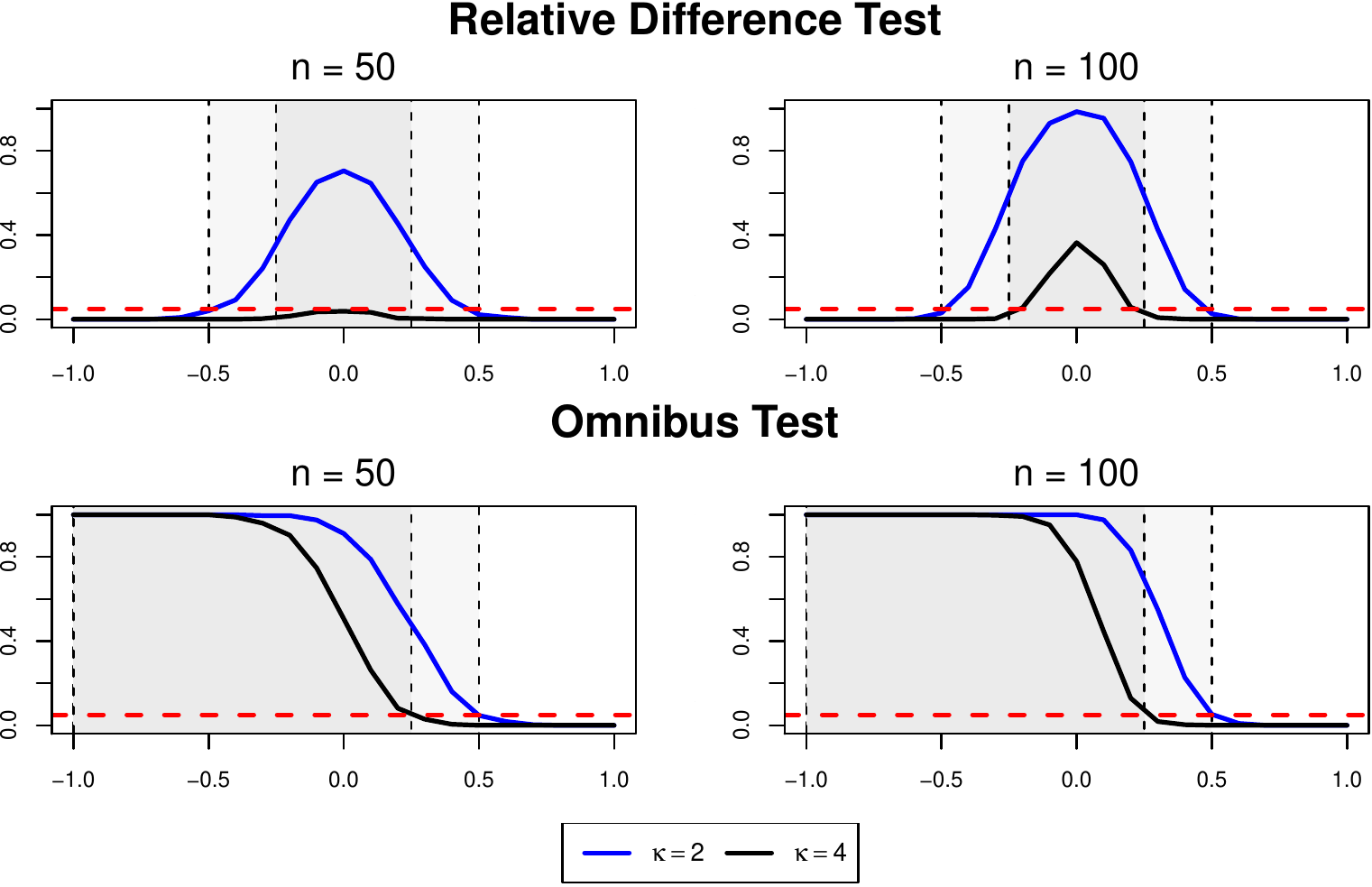}
\caption{Monte Carlo estimate of rejection probability of the relative difference and omnibus likelihood ratio tests.  The shaded light grey and dark grey areas correspond to sets of $\theta_2$ such that the null hypothesis holds with $\kappa = 2$ and $\kappa = 4$, respectively.  The dashed red line denotes the specified size $\alpha = .05$. A color version of this figure is available in the online version of the article.}
\label{fig:PowerSim}
\end{figure}

In Figure \ref{fig:KappaMaxSim}, we plot the quantiles of the $\kappa_{\max}^{\alpha}$ values from 1000 synthetic data sets for each $\theta_{2}$.
We expect that for a fixed $\theta_2$, most $\kappa_{\max}^{\alpha}$ should approach but not exceed $\theta_1/|\theta_2|$ as sample size increases; our simulations are consistent with this expectation.
We note that when the sample size is small, $\kappa^\alpha_{\max}$ tends to underestimate the relative difference in effect size.

\begin{figure}[!h]
\centering
\includegraphics{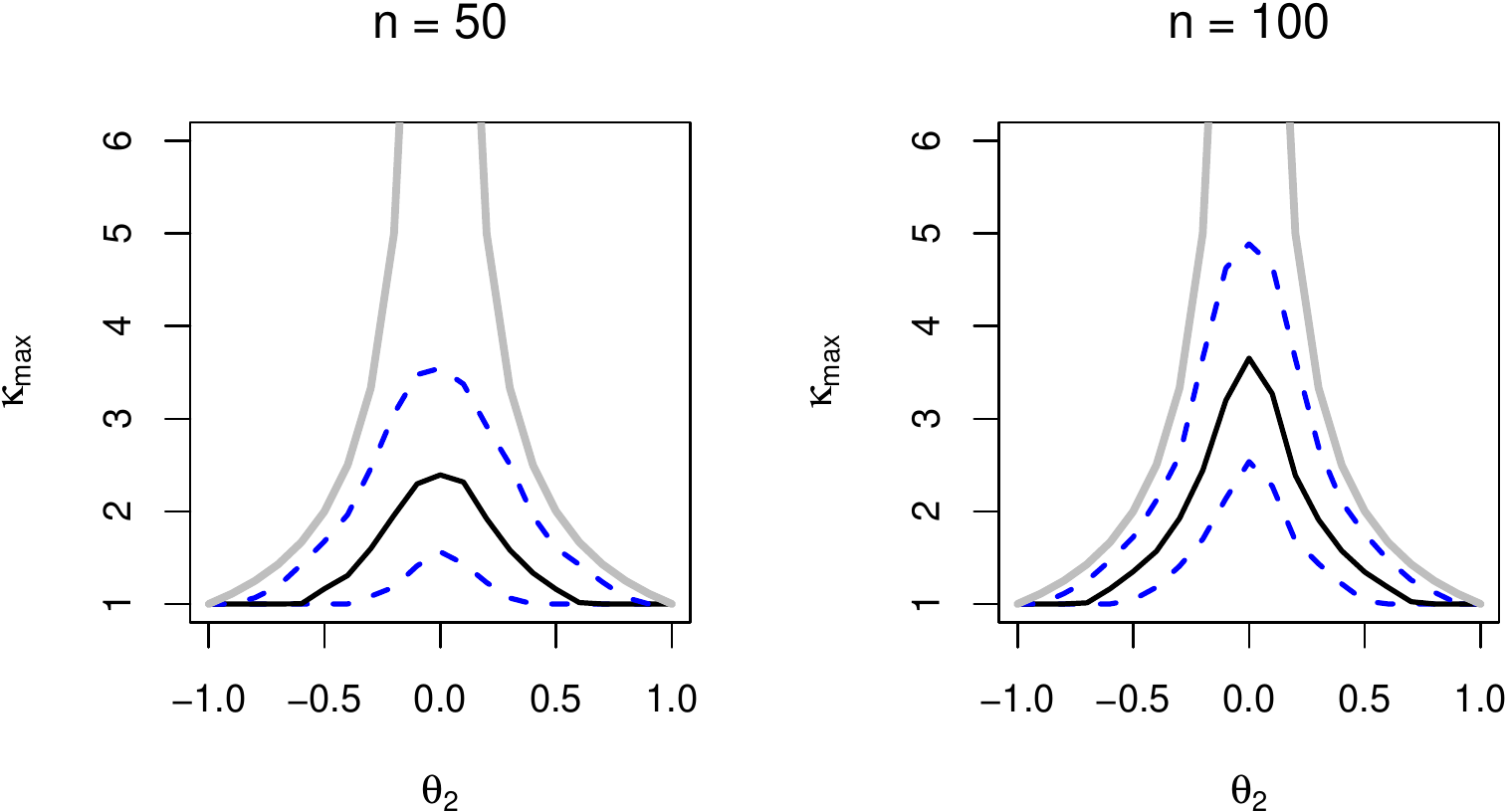}
\caption{Distribution of $\kappa_{\max}^\alpha$ calculated on synthetic data.  The .10, .50, and .90 quantiles are displayed in.  The grey curve represents $|\theta_1|/|\theta_2|$, the value $\kappa_{\max}^\alpha$ is expected to approach for a given $\theta_2$. A color version of this figure is available in the online version of the article.}
\label{fig:KappaMaxSim}
\end{figure}

\section{Data Example}

In this example, we investigate genetic differences in breast cancer sub-types.
Classification of breast cancer based on expression of estrogen receptor (ER) is known to be associated with clinical outcomes.
Approximately 70\% of breast cancers are estrogen receptor positive (ER+) cancers, meaning that estrogen causes cancer cells to grow \citep{lumachi2013treatment}; breast cancers are otherwise estrogen receptor negative (ER-). 
Patients with ER+ breast cancer tend to experience better clinical outcomes than ER- patients \citep{carey2006race}.

We conduct an analysis using publicly available data from The Cancer Genome Atlas (TCGA) \citep{weinstein2013cancer}.
We use clinical data and gene expression data from a total of 806 ER+ patients and 237 ER- patients.

We first investigate the differences between the genetic networks in ER+ and ER- breast cancer.
Both ER+ and ER- breast cancer are expected to have similar pathways, but identifying differences between them may be key to understanding the underlying disease mechanism.
We then conduct an analysis to assess whether any genes in a set known to be associated with breast cancer are strongly prognostic of disease outcomes in only one of the estrogen receptor groups.

\subsection{Differential Network Analysis}

Our objective is to determine whether there are any pairs of genes that are much more strongly associated in one estrogen receptor group than the other.
We consider the set of $p = 145$  genes in the Kyoto Encyclopedia of Genes and Genomes (KEGG) \citep{kanehisa2000kegg} breast cancer pathway and measure the association between gene expression levels using the Pearson correlation.

We test the relative difference null hypothesis for each pair of genes with $\kappa = 2$.
In Figure \ref{fig:Data}, we display the pairs of genes that are statistically significant at the $\alpha = .10$ level after a Bonferroni adjustment.
We find that each of the genes progesterone (PGR), insulin-like growth factor 1 (IGF1R), and  estrogen receptor 1 (ESR1) have multiple differential connections; each belongs to at least two pairs such that the association is twice as strong in the ER+ population than in the ER- population.
These genes have been shown in the literature to be associated with sub-type and prognosis \citep{farabaugh2015role, reinert2019association, kurozumi2017power}.

%\begin{figure}
%\centering
%\includegraphics{"/Users/awhudson/Dropbox/MethodsResearch/QualitativeInteractions/Code/TCGA Analysis/TCGA-Analysis_files/figure-latex/Network2-1".pdf}
%\caption{Pairs of genes that genes in the KEGG breast cancer pathway, for which we reject the relative difference hypothesis with $\kappa = 2$.  Blue edges indicate associations that are stronger in the ER+ group, and red edges indicate associations that are stronger in the ER- group.}
%\label{fig:TCGANet}
%\end{figure}

\subsection{Prognostic Value of Biomarkers}

The goal of this analysis is to assess whether any of the KEGG genes have a stronger association with time to death in one estrogen receptor group than in the other.
For each gene, we fit a univariate Cox proportional-hazards model with time to death as the outcome in both of the estrogen receptor groups separately; we measure association using the log hazard ratio.
A total of 64 deaths occurred in the ER+ group, and 33 deaths occurred the ER- group.
We calculate $\kappa^\alpha_{\max}$ with $\alpha = .10$ for each gene.

In Figure \ref{fig:Data}, we compare the log hazard ratios of the ER+ and ER- groups in a scatterplot.
Though the log hazard ratios for most genes are similar between subgroups,  there are twelve genes with $\kappa^\alpha_{\max}$ larger than one.
A complete list is available in Table 1.
The two genes with the strongest interactions are Growth Factor Receptor-bound Protein 2 (GRB2; $\kappa^{\alpha}_{\max} = 2.04$), which has a stronger association in the ER- group, and Adenomatous Polyposis Coli (APC; $\kappa^{\alpha}_{\max} = 1.91$), which has a stronger association in the ER+ group.
Both genes have been hypothesized to be associated with breast cancer carcinogenesis \citep{daly1994overexpression, jin2001adenomatous}.

\begin{table}[ht]
\centering
\begin{tabular}{lllr}
  \hline
Gene & ER+ Log HR (SE) & ER- Log HR (SE) & $\kappa^{\alpha}_{\max}$ \\ 
  \hline
GRB2 & -0.06 (0.31) & -1.66 (0.68) & 2.04 \\ 
  APC & 1.34 (0.32) & -0.09 (0.33) & 1.91 \\ 
  BAX & -1.05 (0.24) & 0.04 (0.36) & 1.53 \\ 
  PIK3CA & 1.13 (0.28) & 0.14 (0.32) & 1.51 \\ 
  SOS2 & 1.13 (0.36) & -0.1 (0.37) & 1.33 \\ 
  MAP2K2 & -0.87 (0.27) & 0.03 (0.35) & 1.22 \\ 
  GADD45G & -0.52 (0.13) & -0.07 (0.19) & 1.21 \\ 
  HES5 & 0.02 (0.2) & 0.51 (0.18) & 1.19 \\ 
  WNT2 & -0.36 (0.09) & 0 (0.17) & 1.14 \\ 
  DLL4 & 0.09 (0.2) & 0.68 (0.27) & 1.10 \\ 
  FRAT2 & -1.22 (0.31) & -0.45 (0.29) & 1.08 \\ 
  SOS1 & 1.19 (0.3) & -0.34 (0.42) & 1.01 \\ 
   \hline
\end{tabular}
\caption{KEGG genes that are more strongly associated with time to death in one ER group than the other, i.e., $\kappa^\alpha_{\max} > 1$ with $\alpha = .10$.}
\end{table}

\begin{figure}[!h]
\centering
\begin{subfigure}
  \centering
  \includegraphics[width=.4\linewidth]{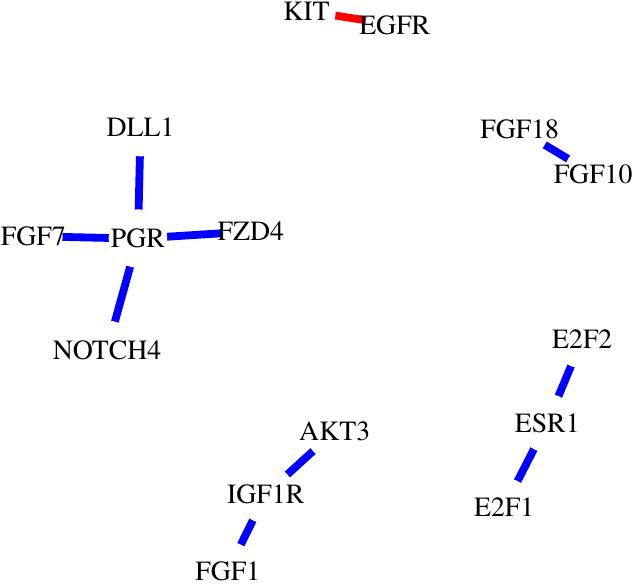}
  \label{fig:a}
\end{subfigure}
\begin{subfigure}
  \centering
  \includegraphics[width=.4\linewidth]{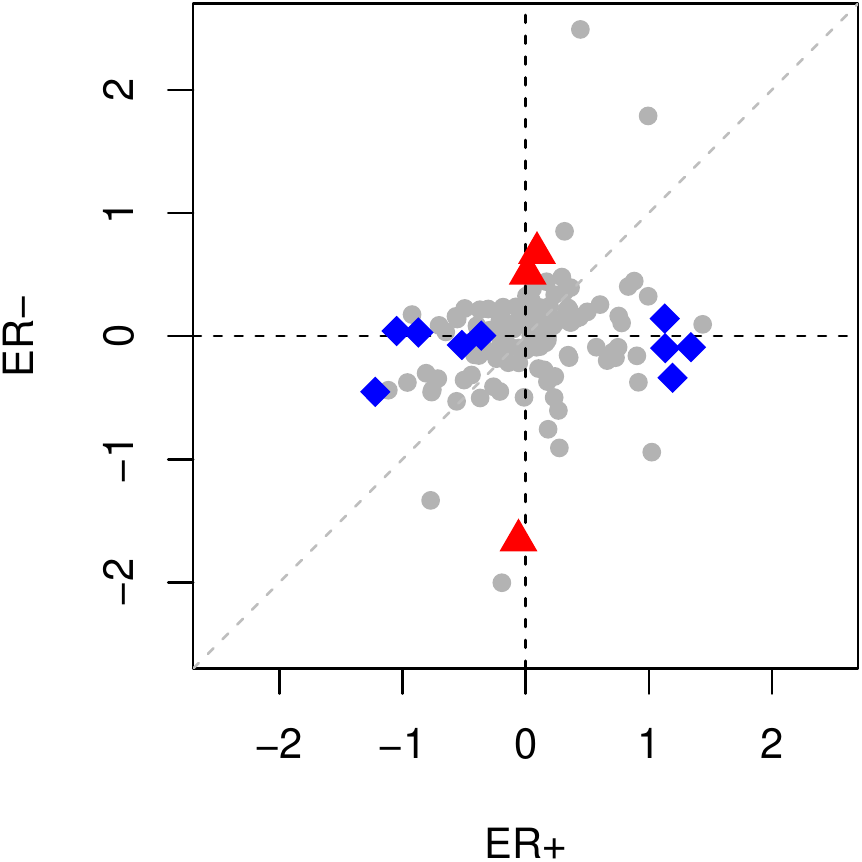}
  \label{fig:b}
\end{subfigure}
\caption{(Left) Pairs of genes in the KEGG breast cancer pathway for which we reject the relative difference hypothesis with $\kappa = 2$.  Blue edges indicate associations that are stronger in the ER+ group, and red edges indicate associations that are stronger in the ER- group. (Right) Log hazard ratios for KEGG genes in ER+ and ER- groups.  The gray dashed line represents the 45-degree line.  Blue diamonds and red triangles indicate genes for which $\kappa^\alpha_{\max} > 1$ with $\alpha = .10$, where the largest log hazard ratio is in the ER+ group and ER- group, respectively. A color version of this figure is available in the online version of the article.}
\label{fig:Data}
\end{figure}

%\begin{figure}
%\centering
%\includegraphics{"/Users/awhudson/Dropbox/MethodsResearch/QualitativeInteractions/Code/TCGA Analysis/TCGA-Analysis_files/figure-latex/Biomarker4-1".pdf}
%\caption{Log hazard ratios for PAM50 genes in ER+ and ER- groups.  The gray dashed line represents the 45-degree line.  Orange and green points indicate genes for which $\kappa^\alpha_{\max} > 1$ in the ER + and ER- groups, respectively, with $\alpha = .1$.}
%\label{fig:TCGABiom}
%\end{figure}

\section{Discussion}

We have proposed a general framework for inference about absence/presence qualitative interactions.
We argued that na\"ive procedures rely upon untenable conditions because the absence/presence hypothesis is ill-posed.
We thus proposed to relax the problem in order to conduct well-calibrated inference that maintains the absence/presence interpretation and only requires mild assumptions.

In simulations, we found that our methodology has low power when signal is weak or sample sizes are small.
To an extent, this is just a feature of the problem; naturally, one would require even more information to detect qualitative interactions than what is required to detect quantitative interactions.
However, we provide no guarantee that our methodology is optimal, as tests for composite hypotheses based upon supremum p-values can be conservative in practice \citep{bayarri2000p}.

Nonetheless, our framework is interpretable and provides a natural approach for quantifying differences in strength of association by sub-population in general settings.
Though we only considered measures of marginal association in our examples, our method can be used with conditional measures of association as well; we only require that asymptotically normal estimates are available.
In particular, our approach remains valid in the high-dimensional setting, where asymptotically normal estimates can be obtained using the techniques of, e.g., \cite{van2014asymptotically} and \cite{zhang2014confidence}.
Finally, our method has good utility in the analysis of genomics data, as we demonstrated in our real data example.

\section{Acknowledgements}

The authors gratefully acknowledge the support of the NSF Graduate Research Fellowship Program under grant DGE-1762114 as well as NSF grant DMS-1561814 and NIH grant R01-GM114029.
Any opinions, findings, and conclusions or recommendations expressed in this material are those of the authors and do not necessarily reflect the views of the funding agencies.

\newpage
\bibliographystyle{biom}
\bibliography{qualint-arxiv}

\section*{Software Availability}

%Additional supporting information may be found in the Supporting Information section at the end of the article.
Implementation of the proposed methodology and code to reproduce results from the data analysis are available at \texttt{https://github.com/awhudson/QualitativeInteractions}.
The TCGA data are available using the public R package \texttt{RTCGA}.

\newpage

\appendix

%  To get the journal style of heading for an appendix, mimic the following.

\section*{Appendix}
%\subsection{Title of appendix}

\subsection*{Theoretical Results for Relative Difference Likelihood Ratio Test}

\setcounter{lemma}{0}
\setcounter{proposition}{0}

Below, we generalize Lemma 1 and Proposition 1 to the setting of unbalanced sample sizes $n_1 \neq n_2$.

\begin{lemma}
 The likelihood ratio test statistic $T^{\kappa}$ can be written as
\[
T^{\kappa} = 
\frac{\hat{\theta}_{\max} - \kappa\hat{\theta}_{\text{min}}}{\sqrt{\hat{\tau}_{\max} + \kappa^2\hat{\tau}_{\min}}},
\]
where $\hat{\tau}_{\max} = n_1^{-1}\sigma^2_1\mathds{1}(|\theta_1| = \theta_{\max}) + n_2^{-1}\sigma^2_2\mathds{1}(|\theta_2| = \theta_{\max})$, and $\hat{\tau}_{\min} = n_1^{-1}\sigma^2_1\mathds{1}(|\theta_1| = \theta_{\min}) + n_2^{-1}\sigma^2_2\mathds{1}(|\theta_2| = \theta_{\min})$. 
\end{lemma}

\begin{proof}[\textbf{Proof of Lemma 1}]
We define $\hat{\theta}_{\text{ord}}$ and $\hat{\Sigma}_{\text{ord}}$ as
\\
\begin{align*}
\hat{\theta}_{\text{ord}} = 
\begin{pmatrix}
\hat{\theta}_{\max} \\ \hat{\theta}_{\min}
\end{pmatrix}
\quad \quad
\hat{\Sigma}_{\text{ord}} = 
\begin{pmatrix}
\hat{\tau}_{\max} & 0
\\
0 & \hat{\tau}_{\min}
\end{pmatrix}.
\end{align*}
\\
Now, we define $\theta^*_{\text{ord}}$ as the projection of $\hat{\theta}_{\text{ord}}$ onto the null region.
If $\hat{\theta}$ is in the null region $\Theta_0^\kappa$, the projection is equal to $\hat{\theta}_{\text{ord}}$.
Otherwise, $\theta^*_{\text{ord}}$ is the projection of $\hat{\theta}_{\text{ord}}$ onto the space spanned by $(\kappa, 1)'$, with distance inversely weighted by $\Sigma_{\text{ord}}$.
We proceed by finding an expression for $\theta^*_{\text{ord}}$.
\begin{align*}
&\theta^*_{\text{ord}}
 =
\\
&\begin{pmatrix}
\kappa \\ 1
\end{pmatrix}
\left[
\begin{pmatrix}
\kappa & 1
\end{pmatrix}
\begin{pmatrix}
\hat{\tau}_{\max} & 0
\\
0 & \hat{\tau}_{\min}
\end{pmatrix}
\begin{pmatrix}
\kappa \\ 1
\end{pmatrix}
\right]^{-1}
\begin{pmatrix}
\kappa & 1
\end{pmatrix}
\begin{pmatrix}
\hat{\tau}_{\max} & 0
\\
0 & \hat{\tau}_{\min}
\end{pmatrix}
\begin{pmatrix}
\hat{\theta}_{\max} \\ \hat{\theta}_{\min}
\end{pmatrix} =
\\
&
\frac{\kappa\hat{\tau}_{\max}\hat{\theta}_{\max} + \hat{\tau}_{\min}\hat{\theta}_{\min}}{\kappa^2\hat{\tau}_{\max} + \hat{\tau}_{\min}}
\begin{pmatrix}
\kappa \\ 1
\end{pmatrix}.
\end{align*}

Now, the weighted distance between $\hat{\theta}_{\text{ord}}$ and $\theta^*_{\text{ord}}$ is
\\
\begin{align*}
&(\hat{\theta}_{\text{ord}} - \theta^*_{\text{ord}})'
\Sigma_{\text{ord}}^{-1}
(\hat{\theta}_{\text{ord}} - \theta^*_{\text{ord}}) =
\\
& \hat{\tau}^{-1}_{\max}
\left(\hat{\theta}_{\max} - \frac{\kappa^2\hat{\tau}_{\max}\hat{\theta}_{\max} + \kappa\hat{\tau}_{\min}\hat{\theta}_{\min}}{\kappa^2\hat{\tau}_{\max} + \hat{\tau}_{\min}} \right)^2
 +
\hat{\tau}^{-1}_{\min}
 \left( \hat{\theta}_{\min} - \frac{\kappa\hat{\tau}_{\max}\hat{\theta}_{\max} + \hat{\tau}_{\min}\hat{\theta}_{\min}}{\kappa^2\hat{\tau}_{\max} + \hat{\tau}_{\min}} \right)^2 =
\\
&
\hat{\tau}_{\text{max}}^{-1}\left(\frac{\hat{\tau}_{\min}^{-1}\hat{\theta}_{\text{max}} - \kappa\hat{\tau}_{\min}^{-1}\hat{\theta}_{\min}}{\kappa^2\hat{\tau}_{\text{max}}^{-1} + \hat{\tau}_{\min}^{-1}}\right)^2 
+
\hat{\tau}_{\min}^{-1}
\left(\frac{\kappa\hat{\tau}_{\text{max}}^{-1}\hat{\theta}_{\text{max}} + \kappa^2\hat{\tau}_{\text{max}}^{-1}\hat{\theta}_{\min}}{\kappa^2\hat{\tau}_{\text{max}}^{-1} + \hat{\tau}_{\min}^{-1}}\right)^2 =
\\
&
\frac{\kappa^2\hat{\tau}^{-2}_{\text{max}}\hat{\tau}_{\min}^{-1} + \hat{\tau}_{\text{max}}^{-1}\hat{\tau}^{-2}_{\min}}{(\kappa^2\hat{\tau}_{\text{max}}^{-1} + \hat{\tau}_{\min}^{-1})^2}
(\hat{\theta}_{\max} - \kappa\hat{\theta}_{\text{min}})^2 =
\\
&
\frac{(\hat{\theta}_{\max} - \kappa\hat{\theta}_{\text{min}})^2}{\hat{\tau}_{\max} + \kappa^2\hat{\tau}_{\min}}.
\end{align*}
\\
Hence, the likelihood ratio test rejects the null for large values of
\\
\[
\frac{(\hat{\theta}_{\max} - \kappa\hat{\theta}_{\text{min}})^2}{\hat{\tau}_{\max} + \kappa^2\hat{\tau}_{\min}}
\mathds{1}
\left(
\hat{\theta}_{\max} - \kappa \hat{\theta}_{\min} > 0 
\right),
\]
or equivalently, for large positive values of ${T}^\kappa$, where
\[
T^\kappa = 
\frac{\hat{\theta}_{\max} - \kappa\hat{\theta}_{\text{min}}}{\sqrt{\hat{\tau}_{\max} + \kappa^2\hat{\tau}_{\min}}}.
\]
\end{proof}

\begin{proposition}
\textit{
Let $n_1 + n_2 = N$, and assume $n_1/N \to \lambda > 0$ as $n_1, n_2 \to \infty$.
Then,
\begin{enumerate}[i]
\item In the interior of the null region, i.e., when $\theta_{\max} - \kappa \theta_{\min} < 0$, $T^\kappa$ converges in distribution to $-\infty$.
\item At all  nonzero boundary points of the null region, i.e., when $\theta_{\max} = \kappa \theta_{\min} > 0$, $T^\kappa$ converges in distribution to a standard normal random variable.
\end{enumerate}
}
\end{proposition}

\begin{proof}[\textbf{Proof of Proposition 1}]
First we prove (\textit{i}).
Suppose $\theta_{\max} - \kappa\theta_{\min}  = b< 0$.  
Then consistency of $\hat{\theta}$ and the continuous mapping theorem imply that $\hat{\theta}_{\max} - \kappa\hat{\theta}_{\min} \to_p b$.  
Now, because $n_1^{-1}\sigma_1^2$ and $n_2^{-1}\sigma_2^2$ both tend to zero, $\frac{1}{\sqrt{\hat{\tau}_{\max} + \kappa^2\hat{\tau}_{\min}}} \to_p \infty$.
Therefore,  $T^\kappa \to_p -\infty.$
\\ \\
Now we prove (\textit{ii}).
Suppose $\theta_{\max} = \kappa\theta_{\min} > 0$.
By applying the delta method,
\begin{align*}
&\sqrt{\frac{n_1n_2}{N}}\left(\hat{\theta}_{\max} - \kappa\hat{\theta}_{\min}\right)
= 
\\
&\sqrt{\frac{n_1n_2}{N}}
\left[
\left(
\max\left\{|\hat{\theta}_1|, |\hat{\theta}_2|\right\} -
\kappa\min\left\{|\hat{\theta}_1|, |\hat{\theta}_2|\right\}
\right) 
-
\left(
\max\left\{|\theta_1|, |\theta_2|\right\} -
\kappa\min\left\{|\theta_1|, |\theta_2|\right\}
\right) 
\right]
=
\\
&\sqrt{\frac{n_1n_2}{N}}
M
\begin{pmatrix}
\hat{\theta}_1 - \theta_1
\\
\hat{\theta}_2 - \theta_2
\end{pmatrix} + o_p(1)
,
\end{align*}
where the matrix $M$ is
\[
M = 
\begin{pmatrix}
(-\kappa)^{\mathds{1}(|\theta_1| = \theta_{\min})}\,\text{sign}(\theta_1) & 0
\\
0 & (-\kappa)^{\mathds{1}(|\theta_2| = \theta_{\min})}\,\text{sign}(\theta_2)
\end{pmatrix}.
\]
Now, Slutsky's theorem implies that
\begin{align*}
\sqrt{\frac{n_1n_2}{N}}
M
\begin{pmatrix}
\hat{\theta}_1 - \theta_1
\\
\hat{\theta}_2 - \theta_2
\end{pmatrix} + o_p(1) \to_d
N\left(0, (1-\lambda)\left(\kappa^2\right)^{\mathds{1}(|\theta_1| = \theta_{\min})}{\sigma^2_1} + \lambda\left(\kappa^2\right)^{\mathds{1}(|\theta_2| = \theta_{\min})}{\sigma^2_2}\right).
\end{align*}
The test statistic $T^\kappa$ can be written as
\begin{align*}
T^\kappa &= 
\frac{\hat{\theta}_{\max} - \kappa\hat{\theta}_{\text{min}}}{\sqrt{\hat{\tau}_{\max} + \kappa^2\hat{\tau}_{\min}}}
\\
&= \frac{\sqrt{\frac{n_1n_2}{N}}\left(\hat{\theta}_{\max} - \kappa\hat{\theta}_{\min}\right)}{\sqrt{\frac{n_1n_2}{N}\hat{\tau}^2_{\max} + \frac{n_1n_2}{N}\kappa^2\hat{\tau}^2_{\min}}},
\end{align*}
with $\hat{\tau}_{\max}$ and $\hat{\tau}_{\min}$ as defined in Proposition 1.
By the continuous mapping theorem, the denominator converges in probability to 
\[
\sqrt{
(1-\lambda)\left(\kappa^2\right)^{\mathds{1}(|\theta_1| = \theta_{\min})}{\sigma^2_1} + \lambda\left(\kappa^2\right)^{\mathds{1}(|\theta_2| = \theta_{\min})}{\sigma^2_2}
}
.
\]
Thus, Slutsky's theorem implies that
\[
T^\kappa \to_d N(0, 1).
\]

\end{proof}

In Proposition 2, we characterize the local asymptotic behavior of the relative difference likelihood ratio test.

\begin{proposition}
\textit{
Let $n_1 + n_2 = N$, and assume $n_1/N \to \lambda > 0$ as $n_1, n_2 \to \infty$.
Suppose $\theta_1 = n^{-1/2}_1c_1$, and $\theta_2 = {n_2^{-1/2}c_2}$.
Further, assume that $\hat{\theta}$ is locally regular in the sense that $\sqrt{n_g}\hat{\theta}_g \to_d N(c_g, {\sigma^2_g})$ for $g \in \{1,2\}$. 
Then as $n_1, n_2 \to \infty$, for all $t > 0$,  $\mathbb{P}\left(T^\kappa > t\right)$ converges to
\begin{align*}
\mathbb{P}(W_{11} > t, W_{12} > t) + \mathbb{P}(W_{11} < -t, W_{12} < -t) + \mathbb{P}(W_{21} > t, W_{22} > t) + \mathbb{P}(W_{21} < -t, W_{22} < -t),
\end{align*}
where $(W_{11},W_{12})$ follows a bivariate normal distribution with mean $(\tilde{c}_{11}, \tilde{c}_{12})$, unit variance and correlation $\nu_1$, and $(W_{21},W_{22})$ follows a bivariate normal distribution with mean $(\tilde{c}_{21}, \tilde{c}_{22})$, unit variance and correlation $\nu_2$, with $\tilde{c}_{11}, \tilde{c}_{12}, \tilde{c}_{21}, \tilde{c}_{22}, \nu_1, \nu_2$ are defined as
\begin{align*}
\tilde{c}_{11} = \frac{(1-\lambda)c_1 - \kappa\lambda c_2}{\sqrt{(1-\lambda)\sigma_1^{2} + \kappa^2\lambda\sigma_2^2}}
\\
\tilde{c}_{12} = \frac{(1-\lambda)c_1 + \kappa\lambda c_2}{\sqrt{(1-\lambda)\sigma_1^{2} + \kappa^2\lambda\sigma_2^2}}
\\
\tilde{c}_{21} = \frac{\lambda c_2 -\kappa(1-\lambda)c_1 }{\sqrt{\kappa^2(1-\lambda)\sigma_1^{2} + \lambda\sigma_2^2}}
\\
\tilde{c}_{22} = \frac{\lambda c_2 + \kappa(1-\lambda)c_1}{\sqrt{\kappa^2(1-\lambda)\sigma_1^{2} + \lambda\sigma_2^2}}
\\
\nu_1 = \frac{(1-\lambda)\sigma_1^2 - \kappa^2\lambda\sigma_2^2}{(1-\lambda)\sigma_1^{2} + \kappa^2\lambda\sigma_2^2 }
\\
\nu_2 = \frac{\lambda\sigma_2^2 - \kappa^2(1-\lambda)\sigma_1^2}{\kappa^2(1-\lambda)\sigma_1^{2} + \lambda\sigma_2^2 }
.
\end{align*}
}
\end{proposition}

\begin{remark}
Setting $c_1 = c_2 = 0$, we obtain the limiting distribution of $T^{\kappa}$ when $\theta_1 = \theta_2 = 0$, which is critical for calculating the p-value $\rho^{\kappa}(t)$.
\end{remark}

\begin{proof}[\textbf{Proof of Proposition 2}]
We first re-write the tail probability as
\begin{align*}
&\mathbb{P}\left(\frac{\hat{\theta}_{\max} - \kappa\hat{\theta}_{\min}}{\hat{\tau}_{\max} + \kappa^2\hat{\tau}_{\min}} > t\right) =
\\ \\
&\mathbb{P}\left(\frac{\hat{\theta}_1 - \kappa\hat{\theta}_2}{\sqrt{n_1^{-1}\sigma_1^2 + \kappa^2n_2^{-1}\sigma_2^2}}  > t, \frac{\hat{\theta}_1 + \kappa\hat{\theta}_2}{\sqrt{n_1^{-1}\sigma_1^2 + \kappa^2n_2^{-1}\sigma_2^2}}  > t\right) + 
\\
&\mathbb{P}\left(\frac{\hat{\theta}_1 - \kappa\hat{\theta}_2}{\sqrt{n_1^{-1}\sigma_1^2 + \kappa^2n_2^{-1}\sigma_2^2}}  < -t, \frac{\hat{\theta}_1 + \kappa\hat{\theta}_2}{\sqrt{n_1^{-1}\sigma_1^2 + \kappa^2n_2^{-1}\sigma_2^2}}  < -t\right) +
\\ 
&\mathbb{P}\left(\frac{\hat{\theta}_2 - \kappa\hat{\theta}_1}{\sqrt{\kappa^2n_1^{-1}\sigma_1^2 + n_2^{-1}\sigma_2^2}}  > t, \frac{\hat{\theta}_2 + \kappa\hat{\theta}_1}{\sqrt{\kappa^2n_1^{-1}\sigma_1^2 + n_2^{-1}\sigma_2^2}}  > t\right) +
\\
&\mathbb{P}\left(\frac{\hat{\theta}_2 - \kappa \hat{\theta}_1}{\sqrt{\kappa^2n_1^{-1}\sigma_1^2 + n_2^{-1}\sigma_2^2}}  < -t, \frac{\hat{\theta}_2 + \kappa \hat{\theta}_1}{\sqrt{\kappa^2n_1^{-1}\sigma_1^2 + n_2^{-1}\sigma_2^2}}  < -t\right).
\end{align*}
\\
Thus,
\begin{align*}
\sqrt{\frac{n_1n_2}{N}}
\begin{pmatrix}
\hat{\theta} _1 \\ \hat{\theta}_2
\end{pmatrix} \to_d
N\left(
\begin{pmatrix}
(1-\lambda)c_1 \\ \lambda c_2
\end{pmatrix}
,
\begin{pmatrix}
(1-\lambda)\sigma_1^2 & 0
\\
0 & \lambda \sigma_2^2
\end{pmatrix}
\right).
\end{align*}
And,
\begin{align*}
\sqrt{\frac{n_1n_2}{N}}
\begin{pmatrix}
1 & -\kappa
\\
1 & \kappa
\end{pmatrix}
\begin{pmatrix}
\hat{\theta} _1 \\ \hat{\theta}_2
\end{pmatrix} \to_d
N\left(
\begin{pmatrix}
(1-\lambda)c_1 - \kappa\lambda c_2  \\ (1-\lambda)c_1 + \kappa\lambda c_2 
\end{pmatrix},
\begin{pmatrix}
(1-\lambda)\sigma_1^2 + \kappa^2\lambda\sigma_2^2 &
(1-\lambda)\sigma_1^2 - \kappa^2 \lambda \sigma_2^2
\\
(1-\lambda)\sigma_1^2 - \kappa^2 \lambda \sigma_2^2 &
\kappa^2(1-\lambda)\sigma_1^2 + \lambda \sigma_2^2
\end{pmatrix}
\right).
\end{align*}

Now,
\begin{align*}
\begin{pmatrix}
\frac{\hat{\theta}_1 - \kappa \hat{\theta}_2}{\sqrt{n_1^{-1}\sigma_1^{2} + \kappa^2n_2^{-1}\sigma_2^2}}
\\
\frac{\hat{\theta}_1 + \kappa\hat{\theta}_2}{\sqrt{n_1^{-1}\sigma_1^{2} + \kappa^2n_2^{-1}\sigma_2^2}}
\end{pmatrix} 
&=
\begin{pmatrix}
\frac{\sqrt{\frac{n_1n_2}{N}}(\hat{\theta}_1 - \kappa \hat{\theta}_2)}{\sqrt{\kappa^2n_2N^{-1}\sigma_1^{2} + n_1N^{-1}\sigma_2^2}}
\\
\frac{\sqrt{\frac{n_1n_2}{N}}(\kappa\hat{\theta}_1 - \hat{\theta}_2)}{\sqrt{\kappa^2n_2N^{-1}\sigma_1^{2} + n_1^{-1}N\sigma_2^2}}
\end{pmatrix} 
\to_d N\left(
\begin{pmatrix}
\tilde{c}_{11} \\ \tilde{c}_{12}
\end{pmatrix}, 
\begin{pmatrix}
1 & \nu_1 \\ \nu_1 & 1
\end{pmatrix}
\right).
\end{align*}
Similarly,
\begin{align*}
\begin{pmatrix}
\frac{\hat{\theta}_2 - \kappa \hat{\theta}_1}{\sqrt{\kappa^2n_1^{-1}\sigma_1^{2} + n_2^{-1}\sigma_2^2}}
\\
\frac{\hat{\theta}_2 + \kappa\hat{\theta}_1}{\sqrt{\kappa^2n_1^{-1}\sigma_1^{2} + n_2^{-1}\sigma_2^2}}
\end{pmatrix} 
\to_d N\left(
\begin{pmatrix}
\tilde{c}_{21} \\ \tilde{c}_{22}
\end{pmatrix}, 
\begin{pmatrix}
1 & \nu_2 \\ \nu_2 & 1
\end{pmatrix}
\right).
\end{align*}

\end{proof}

\subsection*{Theoretical Results for Simultaneous Test for Qualitative Interactions}

In what follows, we generalize Lemma 2 and Proposition 3 to the case of unbalanced sample sizes.
%We propose a likelihood ratio test for $H_0^{\text{P/N},\kappa}:$ $H_0^{\text{P/N}}$ and $H_0^\kappa$ both hold.

\begin{lemma}
The likelihood ratio statistic $T^{\text{\normalfont{P/N}}, \kappa}$ can be written as
\[
T^{\text{\normalfont{P/N}},\kappa} =
\min\left\{
\frac{\left(\hat{\theta}_1 - \kappa \hat{\theta}_2\right)^2}{n_1^{-1}\sigma_1^2 + \kappa^2n_2^{-1}\sigma^2_2}
,
\frac{\left(\kappa \hat{\theta}_1 - \hat{\theta}_2\right)^2}{\kappa^2n_1^{-1}\sigma_1^{2} + n_2^{-1}\sigma_2^2}
\right\}
\mathds{1}\left(\hat{\theta} \in \Theta_1^{\textnormal{P/N},\kappa}\right).
\]
\end{lemma}

\begin{proof}[\textbf{Proof of Lemma 2}]
If $\hat{\theta}$ belongs to the null region, $\Theta_0^{\text{P/N},\kappa}$, the likelihood ratio statistic is clearly zero.
If $\hat{\theta}_1 > \kappa\hat{\theta}_2 > 0$ or $\hat{\theta}_1 < \kappa \hat{\theta}_2 < 0$, the likelihood ratio statistic is minimum of the distances between $\hat{\theta}$ and its projections onto the span of $(1,\kappa)'$ and the span of $(-\kappa, -1)'$.
Algebra (similar to the proof of Lemma 1) gives the desired result, and similarly if $\hat{\theta}_1 > -\kappa\hat{\theta}_2 > 0$ or $\hat{\theta}_1 < -\kappa \hat{\theta}_2 < 0$.
\end{proof}

\begin{proposition}
\textit{Let $n_1 + n_2 = N$, and assume $n_1/N \to \lambda > 0$. 
Then,
\begin{enumerate}[i]
\item If $\theta$ belongs to the interior of the null region $\Theta_0^{\textnormal{P/N},\kappa}$, $T^{\textnormal{P/N}, \kappa}$ converges in distribution to zero.
\item If $\theta$ is on the boundary of the null region, but $\theta \neq 0$, $\mathbb{P}(T^{\textnormal{P/N}, \kappa} > t) \to \frac{1}{2}\mathbb{P}\left(\chi^2_1 > t\right)$ as $n_1, n_2 \to \infty$.
\end{enumerate}
}
\end{proposition}

\begin{proof}[\textbf{Proof of Proposition 3}]
We first prove (\textit{i}).
If $\theta$ belongs to the interior of the null region, consistency of $\hat{\theta}$ and the continuous mapping theorem imply that $\mathds{1}\left(\hat{\theta} \in \Theta_1^{\text{P/N}, \kappa}\right)$ converges in probability to zero.
Then, by Slutsky's theorem, $T^{\text{P/N},\kappa}$ converges in distribution to zero.
\\ \\
To prove \textit{(ii)}, recall that we can write the null and alternative regions as
\begin{align*}
\Theta_0^{\text{P/N}, \kappa} =
&\{0 < \kappa^{-1}\theta_1 < \theta_2 < \kappa \theta_1 \} \cup \{\kappa\theta_1 < \theta_2 < \kappa^{-1} \theta_1 < 0\}
\\
\Theta_1^{\text{P/N}, \kappa} = 
&\left\{\{\theta_1 > \kappa\theta_2\} \cap \{\theta_1 > 0\}\right\} \cup
\left\{\{-\theta_1 > -\kappa\theta_2\} \cap \{\theta_1 < 0\}\right\} \cup
\\
&\left\{\{\theta_2 > \kappa\theta_1\} \cap \{\theta_2 > 0\}\right\} \cup
\left\{\{-\theta_2 > -\kappa\theta_1\} \cap \{\theta_2 < 0\}\right\}.
\end{align*}
Now, we write
\begin{align*}
&\mathbb{P}\left(T^{\text{P/N},\kappa} > t \right) = 
\\ \\
&\mathbb{P}\left(
\frac{\left(\hat{\theta}_1 - \kappa \hat{\theta}_2\right)^2}{n_1^{-1}\sigma_1^{2} + \kappa^2n_2^{-1}\sigma_2^2}\mathds{1}\left(\hat{\theta} \in \Theta_1\right) > t,
\frac{\left(\kappa\hat{\theta}_1 - \hat{\theta}_2\right)^2}{\kappa^2n_1^{-1}\sigma_1^{2} + n_2^{-1}\sigma_2^2}\mathds{1}\left(\hat{\theta} \in \Theta_1\right) > t
\right) =
\\ \\
%&P\Bigg(
%\left\{\frac{\hat{\theta}_1 - \kappa \hat{\theta}_2}{\sqrt{n_1^{-1}\sigma_1^{2} + \kappa^2n_2^{-1}\sigma_2^2}}
%\mathds{1}\left(\hat{\theta} \in \Theta_1\right) > \sqrt{t} \right\} 
%\cup
%\left\{\frac{\hat{\theta}_1 - \kappa \hat{\theta}_2}{\sqrt{n_1^{-1}\sigma_1^{2} + \kappa^2n_2^{-1}\sigma_2^2}}\mathds{1}\left(\hat{\theta} \in \Theta_1\right) < -\sqrt{t}\right\},
%\\
%&
%\quad\,\,\,
%\left\{\frac{\kappa\hat{\theta}_1 - \hat{\theta}_2}{\sqrt{\kappa^2n_1^{-1}\sigma_1^{2} + n_2^{-1}\sigma_2^2}}\mathds{1}\left(\hat{\theta} \in \Theta_1\right) > \sqrt{t}\right\} 
%\cup
%\left\{\frac{\kappa\hat{\theta}_1 - \hat{\theta}_2}{\sqrt{\kappa^2n_1^{-1}\sigma_1^{2} + n_2^{-1}\sigma_2^2}}\mathds{1}\left(\hat{\theta} \in \Theta_1\right) > -\sqrt{t}\right\}
%\Bigg) =
%\\
%\\
&\mathbb{P}\left(
\frac{\hat{\theta}_1 - \kappa \hat{\theta}_2}{\sqrt{n_1^{-1}\sigma_1^{2} + \kappa^2n_2^{-1}\sigma_2^2}}\mathds{1}\left(\hat{\theta} \in \Theta_1\right)  > \sqrt{t},
\frac{\kappa\hat{\theta}_1 - \hat{\theta}_2}{\sqrt{\kappa^2n_1^{-1}\sigma_1^{2} + n_2^{-1}\sigma_2^2}}\mathds{1}\left(\hat{\theta} \in \Theta_1\right)  > \sqrt{t} 
\right)  +
\\
&\mathbb{P}\left(
\frac{\hat{\theta}_1 - \kappa \hat{\theta}_2}{\sqrt{n_1^{-1}\sigma_1^{2} + \kappa^2n_2^{-1}\sigma_2^2}}\mathds{1}\left(\hat{\theta} \in \Theta_1\right)  > \sqrt{t}
,
\frac{\kappa\hat{\theta}_1 - \hat{\theta}_2}{\sqrt{\kappa^2n_1^{-1}\sigma_1^{2} + n_2^{-1}\sigma_2^2}}\mathds{1}\left(\hat{\theta} \in \Theta_1\right)  < -\sqrt{t}
\right) +
\\
&\mathbb{P}\left(
\frac{\hat{\theta}_1 - \kappa \hat{\theta}_2}{\sqrt{n_1^{-1}\sigma_1^{2} + \kappa^2n_2^{-1}\sigma_2^2}}\mathds{1}\left(\hat{\theta} \in \Theta_1\right)  < -\sqrt{t}
,
\frac{\kappa\hat{\theta}_1 - \hat{\theta}_2}{\sqrt{\kappa^2n_1^{-1}\sigma_1^{2} + n_2^{-1}\sigma_2^2}}\mathds{1}\left(\hat{\theta} \in \Theta_1\right)  > \sqrt{t} 
\right) +
\\
&\mathbb{P}\left(
\frac{\hat{\theta}_1 - \kappa \hat{\theta}_2}{\sqrt{n_1^{-1}\sigma_1^{2} + \kappa^2n_2^{-1}\sigma_2^2}}\mathds{1}\left(\hat{\theta} \in \Theta_1\right)  < -\sqrt{t} 
,
\frac{\kappa\hat{\theta}_1 - \hat{\theta}_2}{\sqrt{\kappa^2n_1^{-1}\sigma_1^{2} + n_2^{-1}\sigma_2^2}}\mathds{1}\left(\hat{\theta} \in \Theta_1\right)  < -\sqrt{t} 
\right).
\end{align*}
\\
The second and third summands above are exactly equal to 0.
 To see this, note that in the second term $\hat{\theta}_1 > \kappa \hat{\theta}_2$ and $\hat{\theta}_1 < \kappa^{-1}\hat{\theta}_2$ implies that $\hat{\theta} \in \Theta_0^{\text{P/N},\kappa}$, and similarly for the third term.
In the first and last summands, we can ignore the term $\mathds{1}\left(\hat{\theta} \in \Theta_1^{\text{P/N},\kappa}\right)$. 
To see this, note that in the first term $\hat{\theta}_1 > \kappa \hat{\theta}_2$ and $\hat{\theta_1} > \kappa^{-1}\hat{\theta_2}$ imply that $\hat{\theta} \in \Theta_1^{\text{P/N},\kappa}$; a similar argument holds for the fourth term.  Thus,
\begin{align}
\mathbb{P}\left(T^{\text{P/N},\kappa} > t \right) &=
\mathbb{P}\left(
\frac{\hat{\theta}_1 - \kappa \hat{\theta}_2}{\sqrt{n_1^{-1}\sigma_1^{2} + \kappa^2n_2^{-1}\sigma_2^2}}  > \sqrt{t},
\frac{\kappa\hat{\theta}_1 - \hat{\theta}_2}{\sqrt{\kappa^2n_1^{-1}\sigma_1^{2} + n_2^{-1}\sigma_2^2}}  > \sqrt{t} 
\right)  + \nonumber
\\
&\quad\,\,
\mathbb{P}\left(
\frac{\hat{\theta}_1 - \kappa \hat{\theta}_2}{\sqrt{n_1^{-1}\sigma_1^{2} + \kappa^2n_2^{-1}\sigma_2^2}}  < -\sqrt{t}
,
\frac{\kappa\hat{\theta}_1 - \hat{\theta}_2}{\sqrt{\kappa^2n_1^{-1}\sigma_1^{2} + n_2^{-1}\sigma_2^2}}  < -\sqrt{t} 
\right). \label{eqn:OmniTail}
\end{align}
Suppose, for the moment, $\theta_1 = \kappa \theta_2 > 0$.  Then
\begin{align*}
\sqrt{\frac{n_1n_2}{N}}
\begin{pmatrix}
\hat{\theta}_1 - \theta_1
\\
\hat{\theta}_2 - \theta_2
\end{pmatrix} \to_d
N
\left(
0,
\begin{pmatrix}
(1-\lambda)\sigma_1^2  & 0
\\
0 & \lambda \sigma_2^2
\end{pmatrix}
\right).
\end{align*}
Therefore,
\begin{align*}
\sqrt{\frac{n_1n_2}{N}}
\begin{pmatrix}
\hat{\theta}_1 - \kappa\hat{\theta}_2
\\
\kappa\hat{\theta}_1 - \hat{\theta}_2
\end{pmatrix} 
&=
\sqrt{\frac{n_1n_2}{N}}
\begin{pmatrix}
\hat{\theta}_1 - \kappa\hat{\theta}_2
\\
\kappa(\hat{\theta}_1 - \theta_1) - (\hat{\theta}_2 - \theta_2) + (\kappa\theta_1 - \theta_2)
\end{pmatrix} 
\\
&\to_d
N\left(
0, 
\begin{pmatrix}
(1-\lambda)\sigma_1^2 + \kappa^2\lambda\sigma_2^2 &
\kappa(1-\lambda)\sigma_1^2 + \kappa \lambda \sigma_2^2
\\
\kappa(1-\lambda)\sigma_1^2 + \kappa \lambda \sigma_2^2 &
\kappa^2(1-\lambda)\sigma_1^2 + \lambda \sigma_2^2
\end{pmatrix}
\right) +
\begin{pmatrix}
0 \\ \infty
\end{pmatrix}.
\end{align*}
Thus,
\begin{align*}
\mathbb{P}\left(T^{\text{P/N},\kappa} > t\right) &\to 
1 - \Phi\left(\sqrt{t}\right) = \frac{1}{2}\mathbb{P}\left(\chi^2_1 > t \right).
\end{align*}
A similar argument follows if $\theta_1 = \kappa\theta_2 < 0$, $\theta_2 = \kappa\theta_1 > 0$, or $\theta_2 = \kappa \theta_1 < 0$.
\end{proof}

In Proposition 4, we characterize the local asymptotic behavior of the omnibus test for qualitative interactions.

\begin{proposition}
\textit{
Let $n_1 + n_2 = N$, and assume $n_1/N  \to \lambda > 0$ as $n_1, n_2 \to \infty$.  Suppose $\theta_1 = n_1^{-1/2}c_1, \theta_2  n^{-1/2}c_2$ with $c_1, c_2 \geq 0$.  Further, assume $\hat{\theta}$ is locally regular in the sense that $\sqrt{n_g}\hat{\theta}_g \to_d N(c_g, \sigma_g^2)$ for $g \in \{1,2\}$.  Then as $n_1, n_2 \to \infty$, for all $t > 0$,  $\mathbb{P}\left(T^{\textnormal{P/N},\kappa} > t\right)$ converges to
\[
\mathbb{P}\left(V_1 > \sqrt{t}, V_2 > \sqrt{t}\right) + \mathbb{P}\left(V_1 < -\sqrt{t}, V_2 < -\sqrt{t}\right),
\]
where $(V_1,V_2)$ follows a bivariate normal distribution with mean $(\tilde{c}_1, \tilde{c}_2)$, unit variance, and correlation $\nu$, where
\begin{align*}
&\tilde{c}_{1} = \frac{(1-\lambda)c_1 - \kappa \lambda c_2}{\sqrt{(1-\lambda)\sigma_1^{2} + \kappa^2\lambda\sigma_2^2}}
\\
&\tilde{c}_{2} = \frac{\kappa(1-\lambda)c_1 - \lambda c_2}{\sqrt{\kappa^2(1-\lambda)\sigma_1^{2} + \lambda\sigma_2^2}}
\\
&\nu = \frac{\kappa(1-\lambda)\sigma_1^2 + \kappa\lambda\sigma_2^2}{\sqrt{(1-\lambda)\sigma_1^{2} + \kappa^2\lambda\sigma_2^2}\sqrt{\kappa^2(1-\lambda)\sigma_1^{2} + \lambda\sigma_2^2} }.
\end{align*}
}
\end{proposition}

\begin{remark}
Setting $c_1 = c_2 = 0$, we obtain the limiting distribution of $T^{\text{P/N},\kappa}$ when $\theta_1 = \theta_2 = 0$, which is critical for calculating the p-value $\rho^{\text{P/N},\kappa}(t)$.
\end{remark}

\begin{proof}[\textbf{Proof of Proposition 4}]
First, 
\begin{align*}
\sqrt{\frac{n_1n_2}{N}}
\begin{pmatrix}
\hat{\theta} _1 \\ \hat{\theta}_2
\end{pmatrix} \to_d
N\left(
\begin{pmatrix}
(1-\lambda)c_1 \\ \lambda c_2
\end{pmatrix}
,
\begin{pmatrix}
(1-\lambda)\sigma_1^2 & 0
\\
0 & \lambda \sigma_2^2
\end{pmatrix}.
\right)
\end{align*}
Also,
\begin{align*}
\sqrt{\frac{n_1n_2}{N}}
\begin{pmatrix}
1 & -\kappa
\\
\kappa & -1
\end{pmatrix}
\begin{pmatrix}
\hat{\theta} _1 \\ \hat{\theta}_2
\end{pmatrix} \to_d
N\left(
\begin{pmatrix}
(1-\lambda)c_1 - \kappa \lambda c_2 \\ \kappa(1-\lambda)c_1 - \lambda c_2
\end{pmatrix},
\begin{pmatrix}
(1-\lambda)\sigma_1^2 + \kappa^2\lambda\sigma_2^2 &
\kappa(1-\lambda)\sigma_1^2 + \kappa \lambda \sigma_2^2
\\
\kappa(1-\lambda)\sigma_1^2 + \kappa \lambda \sigma_2^2 &
\kappa^2(1-\lambda)\sigma_1^2 + \lambda \sigma_2^2
\end{pmatrix}
\right).
\end{align*}

Thus,
\begin{align*}
\begin{pmatrix}
\frac{\hat{\theta}_1 - \kappa \hat{\theta}_2}{\sqrt{n_1^{-1}\sigma_1^{2} + \kappa^2n_2^{-1}\sigma_2^2}}
\\
\frac{\kappa\hat{\theta}_1 - \hat{\theta}_2}{\sqrt{\kappa^2n_1^{-1}\sigma_1^{2} + n_2^{-1}\sigma_2^2}}
\end{pmatrix} 
&=
\begin{pmatrix}
\frac{\sqrt{\frac{n_1n_2}{N}}(\hat{\theta}_1 - \kappa \hat{\theta}_2)}{\sqrt{n_2N^{-1}\sigma_1^{2} + \kappa^2n_1N^{-1}\sigma_2^2}}
\\
\frac{\sqrt{\frac{n_1n_2}{N}}(\kappa\hat{\theta}_1 - \hat{\theta}_2)}{\sqrt{\kappa^2n_2N^{-1}\sigma_1^{2} + n_1^{-1}N\sigma_2^2}}
\end{pmatrix} 
\to_d N\left(
\begin{pmatrix}
\tilde{c}_{1} \\ \tilde{c}_{2}
\end{pmatrix}, 
\begin{pmatrix}
1 & \nu \\ \nu & 1
\end{pmatrix}
\right).
\end{align*}
\\
Application of \eqref{eqn:OmniTail} completes the argument.
\end{proof}

\section*{Calculation of $\kappa^\alpha_{\max}$}

Define $t^\kappa$ as a realization of the likelihood ratio statistic $T^\kappa$ for the relative difference hypothesis, calculated on the observed data.
For a pre-specified $\alpha < 1$, recall that $\kappa_{\max}^\alpha \equiv \sup\{\kappa: \kappa > 1, \rho^\kappa(t^\kappa) < \alpha\}$ is defined as the largest $\kappa > 1$ such that the likelihood ratio test rejects the relative difference null.
Defining $F^\kappa(t) \equiv \lim_{n_1, n_2 \to \infty} \mathbb{P}\left(T^\kappa > t|\theta = 0\right)$, we can evaluate $\kappa^\alpha_{\max}$ as 
\begin{align*} 
\kappa^\alpha_{\max} &= \min\left\{
 \sup\left\{\kappa: \kappa > 1, 1-\Phi\left(t^\kappa\right)< \alpha\right\},  
 \sup\{\kappa: \kappa > 1, 1 - F^\kappa(t^\kappa) < \alpha\} \right\}
 \\
 &\equiv \min\{\pi_1, \pi_2\}.
\end{align*}
Therefore, we are only required to calculate $\pi_1$ and $\pi_2$.

Both $\pi_1$ and $\pi_2$ can be calculated with root-finding algorithms.
To see this, we first observe that $t^\kappa$ is monotone in $\kappa$, recalling that the numerator of $T^\kappa = \frac{\hat{\theta}_{\max} - \kappa \hat{\theta}_{\min}}{\hat{\tau}_{\max} + \kappa^2\hat{\tau}_{\min}}$ decreases in $\kappa$ while the denominator increases.
Monotonicity of $\Phi(\cdot)$ implies that $1 - \Phi(t^\kappa)$ is monotone in $\kappa$.
Now,  applying of Proposition 2, 
\[
1 - F^\kappa(t) = 2\left\{\mathbb{P}\left(W_{11} > t, W_{12} > t\right) +  \mathbb{P}\left(W_{21} > t, W_{22} > t \right) \right\},
\]
where $(W_{11}, W_{12})$ and $(W_{21}, W_{22})$ follow bivariate normal distributions with mean zero, unit variance, and correlations $\nu_1$ and $\nu_2$, as defined in Proposition 2, respectively.
Both $\nu_1$ and $\nu_2$ are monotone decreasing in $\kappa$, so Theorem 8 in \cite{muller2001stochastic} implies that $1 - F^\kappa(t)$ is monotone in $t$.
Thus $1 - F^\kappa(t^\kappa)$ is monotone in $\kappa$.

Monotonicity of $1 - \Phi(t^\kappa)$ and $1 - F^\kappa(t^\kappa)$  implies that if the likelihood ratio test rejects the null hypothesis for some $\kappa > 1$, $\pi_1$ and $\pi_2$ are the unique roots of
\begin{align*}
&f_1 (\kappa)= 1 - \Phi(t^\kappa) - \alpha
\\
&f_2(\kappa) = 1 - F^{\kappa}(t^\kappa) - \alpha,
\end{align*}
respectively. These roots can be easily calculated via, e.g., the bisection method.  In our implementation, we use the \texttt{uniroot} function in R.

\label{lastpage}

\end{document}